\documentclass[11pt]{amsart}

\usepackage[a4paper,top=2cm,bottom=2cm,left=2cm,right=2cm]{geometry}

\usepackage[T1]{fontenc}
\usepackage{amsmath,amssymb,amsthm}
\usepackage{bm}
\usepackage{graphicx}
\usepackage{subcaption}
\usepackage{float}
\usepackage{cite}          
\usepackage{hyperref}

\newcommand{\sech}{\mathop{\rm sech}\nolimits}
\newcommand{\csch}{\mathop{\rm csch}\nolimits}

\newcommand{\abs}[1]{\left|#1\right|}

\newtheorem{lemma}{Lemma}[section]

\newtheorem{proposition}{Proposition}[section]
\newtheorem{corollary}{Corollary}[section]

\newtheorem{remark}{Remark}[section]

\theoremstyle{definition}

\numberwithin{equation}{section}

\title[Separatrix Splitting in CC Reductions of $\phi^4$ Kinks]{Separatrix Splitting and Chaotic Dynamics in 
Collective-Coordinate Reductions of Driven $\phi^4$ Kinks}

\author{Vassilios M Rothos}
\address{School of Mechanical Engineering, Faculty of Engineering\\
Aristotle University of Thessaloniki,\\
Thessaloniki 54124, Greece}
\email{rothos@auth.gr}
\thanks{The research was funded by Aristotle University of Thessaloniki (AUTH) Research Council grants number 73191, 73699, 11682.}

\subjclass[2020]{37C29, 34C28, 37D45, 35Q51, 37K40, 70H08}

\keywords{Nonlinear waves; $\phi^4$ model; collective coordinates;
Melnikov theory; separatrix splitting; chaotic dynamics.}
\date{\today}

\dedicatory{}

\begin{document}

\begin{abstract}
We investigate the emergence of chaotic dynamics in collective-coordinate
reductions of a driven and spatially modulated $\phi^4$ field describing the
motion of topological kinks. Focusing on finite-dimensional effective models,
we consider both translation-only and constraint-consistent
two--collective--coordinate reductions in the presence of spatial pinning,
dissipation, and traveling-wave forcing. All overlap integrals generated by the
modulations are evaluated in closed form, and Melnikov theory is used to obtain
an explicit analytical characterization of separatrix splitting together with
sufficient conditions for transverse intersection of the perturbed invariant
manifolds. A central structural result is that the leading back-reaction of the
kink internal mode on the translational coordinate vanishes identically, as a
consequence of the self-adjointness of the linearized operator and the parity of
its translational zero mode; the shape mode can therefore influence the
splitting only at second order in the modulation amplitudes. Treating the
internal-mode response along the separatrix exactly, rather than adiabatically,
shows further that its resonance is filtered out of the splitting integral.
Direct comparison of the two reductions confirms both findings: the shape mode
is strongly and resonantly excited, yet leaves the chaotic dynamics of the kink
centre unchanged. Numerical simulations of the reduced systems, including
stroboscopic Poincar\'e sections, Lyapunov exponent computations and a sweep of
the forcing amplitude across the analytical threshold, confirm the predictions.
\end{abstract}
\maketitle

\section{Introduction}
\label{sec1}

Topological kinks in $\phi^4$ field theory form a canonical class of
coherent structures whose dynamics combines particle-like motion with
intrinsically field-theoretic degrees of freedom.
In homogeneous media, $\phi^4$ kinks propagate as stable solitary waves
and possess a localized internal (shape) mode, a feature absent in
integrable models such as sine--Gordon
\cite{CampbellSchonfeldWingate1983,LizunovaSutcliffe2021}.
When dissipation, external forcing, or spatial inhomogeneities are
introduced, translational invariance is broken and kink dynamics may
exhibit pinning--depinning transitions, resonant energy exchange with
internal modes, and, in suitable regimes, chaotic motion
\cite{Chacon2006,KivsharMalomed1989}. From the viewpoint of nonlinear dynamical systems,
such phenomena are naturally interpreted in terms of
bifurcations, separatrix dynamics, and the onset of
chaotic motion in reduced phase-space descriptions.

A standard analytical approach to such problems is provided by
collective-coordinate (CC) reductions, whereby the infinite-dimensional
field dynamics is projected onto a finite-dimensional system governing a
small set of dynamically relevant degrees of freedom.
Translation-only reductions for $\phi^4$ kinks are well established and
capture basic transport and depinning mechanisms
\cite{KivsharMalomed1989}; for a broad survey of collective-coordinate
approaches to solitary-wave dynamics see \cite{CaputoKevrekidis2020}.
More refined two--collective--coordinate (2-CC) reductions, incorporating
the internal mode, reproduce resonance phenomena, energy transfer between
modes, and bifurcation structures that are inaccessible within
single-coordinate descriptions
\cite{CampbellSchonfeldWingate1983,Manton2021,Weigel2018}.
Recent work has further shown that constraint-consistent multi-degree-of-freedom
reductions remain effective in nonautonomous and spatially modulated
settings, including periodic pinning landscapes and spatiotemporal forcing.

In nonautonomous settings, collective-coordinate reductions generically
lead to finite-dimensional dynamical systems possessing separatrices.
This structure naturally renders Melnikov theory a suitable analytical tool
for detecting separatrix splitting, transverse intersections of stable and
unstable manifolds, and the onset of chaotic dynamics
\cite{HolmesMarsden1982,Wiggins1996}.  Melnikov-type techniques have been widely used to detect separatrix
splitting and chaotic dynamics in perturbed nonlinear systems.
Applications include periodically perturbed oscillators, hybrid
piecewise-smooth dynamical systems, and nonlinear wave equations.
Representative studies in the IJBC literature include
\cite{Dua2015,Li2016,Huang2004,Wu2006}. 

In the present work we exploit this structure to obtain a fully explicit
analytical characterization of separatrix splitting in collective-coordinate
reductions of driven and spatially modulated $\phi^4$ kinks.  
Using Melnikov theory, we derive closed-form conditions for transverse
homoclinic intersections and the onset of chaotic dynamics in the
resulting reduced phase space.
Although Melnikov methods have been employed in a variety of soliton and
kink problems \cite{Chacon2006}, their systematic application to
constraint-consistent collective-coordinate reductions of $\phi^4$ kinks,
particularly in spatially periodic media, remains comparatively limited.

It is worth emphasizing that the viewpoint adopted in the present work
differs from that of recent PDE-centered studies on collective-coordinate
modeling of $\phi^4$ kinks \cite{GatlikDobrowolskiCaputoKevrekidis2026}.
Such studies primarily focus on assessing the quantitative accuracy of
reduced models through systematic comparison with the full field dynamics.
In contrast to recent studies that focus primarily on assessing the
quantitative accuracy of collective-coordinate models with respect to
the full PDE, the present work addresses a different question: the
intrinsic phase-space structure of collective-coordinate reductions.
Rather than evaluating their quantitative fidelity with respect to the
underlying field equation, we investigate the dynamical mechanisms that
give rise to chaotic motion within the reduced models and provide an
explicit analytical characterization of separatrix splitting and chaotic
dynamics in the reduced phase space.  This perspective places the reduced collective-coordinate models
within the broader framework of low-dimensional dynamical systems
arising from nonlinear wave equations.

In this work, we consider a driven--damped $\phi^4$ field with spatially
periodic coefficients and a traveling-wave modulation of the nonlinear
coupling.
Our aim is not to assess the quantitative fidelity of reduced models with
respect to the full partial differential equation, but rather to obtain a
fully explicit analytical characterization of separatrix splitting and
chaotic dynamics at the level of the reduced collective-coordinate systems.
To this end, we derive both translation-only and two--collective--coordinate
reductions and apply Melnikov theory to obtain closed-form criteria for the
existence of chaos in terms of physically meaningful control parameters.

This perspective differs from recent PDE-focused studies
\cite{pgk23,pgk24,pgk25}, which emphasize direct numerical comparison
between reduced models and the underlying field equation in order to assess
their domain of validity.
Instead, we focus on the geometric mechanisms responsible for chaos in the
reduced phase space itself.
In particular, we identify parameter regimes in which chaotic dynamics is
rigorously predicted by Melnikov theory at the level of the reduced systems, and
we delimit the parameter window within which that prediction is justified.

The main analytical contribution of this work is the derivation of explicit
Melnikov conditions for separatrix splitting in collective-coordinate
reductions of driven $\phi^4$ kinks.  

This result provides a transparent dynamical-systems framework for
understanding how chaotic transport mechanisms can arise in reduced
models of driven nonlinear wave equations. 
Our analysis emphasizes three complementary aspects.
First, all overlap integrals induced by spatial and spatiotemporal
modulations are evaluated explicitly, yielding closed-form expressions for
the effective pinning and driving terms.
Second, we construct a constraint-consistent two--collective--coordinate
reduction that incorporates the internal mode and its back-reaction on
translational motion, thereby avoiding ambiguities associated with gauge
freedom. In doing so we establish that this back-reaction vanishes identically
at leading order --- a structural consequence of the self-adjointness of the
linearized operator and the parity of its translational zero mode --- so that
the shape mode influences the kink centre only at second order in the
modulation amplitudes.
Third, we develop a Melnikov analysis based on the adjoint variational
equation in the extended phase space, allowing the separatrix splitting in
the 2-CC system to be decomposed into translation-mediated and
internal-mode-mediated contributions, and we identify both the internal-mode
resonance and a decoupling wavenumber at which the internal channel closes
identically.
Since Melnikov theory applies rigorously to finite-dimensional systems, the
analysis is intentionally confined to the reduced collective-coordinate
dynamics.

Stroboscopic Poincar\'e sections reveal a chaotic set organized about the
unperturbed separatrix and repeatedly crossing it, while computation of the
maximal Lyapunov exponent, together with a sweep of the forcing amplitude across
the analytical threshold, provides quantitative confirmation of deterministic
chaos. Taken together, these results establish a transparent analytical framework
for identifying and interpreting chaotic dynamics in reduced models of
driven $\phi^4$ kinks.

The paper is organized as follows. In Sec.~\ref{sec:model} we introduce the driven
and spatially modulated $\phi^4$ model. Section~\ref{sec_cc}
 derives the
collective-coordinate reductions and evaluates the relevant overlap
integrals in closed form, yielding explicit reduced dynamical systems.
In Sec.~\ref{sec:melnikov} we apply Melnikov theory to analyze separatrix splitting and
the onset of chaotic dynamics in the reduced phase space. Section~\ref{sec:numerics}
presents numerical simulations of the reduced systems, including
Poincar\'e sections and Lyapunov exponent computations that confirm
the analytical predictions. Finally, Sec.~\ref{sec:conclusions} summarizes the results and
discusses possible extensions of the approach.
\section{Driven and spatially modulated \texorpdfstring{$\phi^4$}{phi\textasciicircum 4} model}
\label{sec:model}

We consider a driven--damped $\phi^4$ field theory on the real line,
governed by
\begin{equation}
\partial_t^2 \phi + \eta \partial_t \phi
- \partial_x\!\big(F(x)\partial_x \phi\big)
+ \lambda(t,x)\,\phi(\phi^2-1)
= -\Gamma ,
\label{eq:pde}
\end{equation}
where $\phi(x,t)$ is a real scalar field, $\eta>0$ denotes linear
dissipation, and $\Gamma$ is a constant external bias.
The functions $F(x)$ and $\lambda(t,x)$ introduce spatial and
spatiotemporal modulations of the medium, respectively, rendering the system explicitly nonautonomous and generally nonintegrable.

Motivated by analytical tractability, we consider modulations of the elementary Fourier form,
\begin{equation}
F(x) = 1 + \varepsilon_1 \sin(qx), \qquad
\lambda(t,x) = 1 + \varepsilon_2 \sin(kx - \omega t),
\label{eq:mods}
\end{equation}
where $0<\varepsilon_1,\varepsilon_2 \ll 1$ are small parameters.
The spatial modulation $F(x)$ induces an effective periodic pinning
landscape for the kink center, while the traveling-wave modulation of
the nonlinear coupling $\lambda(t,x)$ provides a time-periodic drive
with both spatial and temporal phase dependence.

We restrict attention to kink-type solutions in the topological sector
connecting the two homogeneous vacua,
\begin{equation*}
\lim_{x\to -\infty} \phi(x,t) = -1, \qquad
\lim_{x\to +\infty} \phi(x,t) = +1,
\end{equation*}
and assume sufficient spatial localization so that all projection and
overlap integrals arising in the collective-coordinate reduction are
well-defined.
In the absence of dissipation, forcing, and modulation
($\eta=\Gamma=\varepsilon_1=\varepsilon_2=0$),
Eq.~\eqref{eq:pde} admits the standard $\phi^4$ kink solution
\begin{equation}
\phi_K(x-X) =
\tanh\!\left(\frac{x-X}{\sqrt{2}}\right),
\label{eq:kink}
\end{equation}
parametrized by its center position $X$.

Throughout this work, we operate in a perturbative regime in which the
modulation amplitudes $\varepsilon_1$, $\varepsilon_2$, the damping
$\eta$, and the bias $\Gamma$ are treated as small parameters in the
sense required by collective-coordinate theory and Melnikov analysis.
In the translation-only reduction, this corresponds to assuming that
the field remains close to a rigid kink profile with a slowly varying
center $X(t)$.
In the refined two--collective--coordinate (2-CC) setting, we
additionally assume a weak-wobbling regime in which the internal-mode
amplitude remains small and enters the translational dynamics only at
leading order.

These requirements can be made quantitative, and it is useful to do so at the
outset, since they delimit the parameter window in which every later statement
is meant to hold. Writing $\Omega_0^2=\varepsilon_1A_1(q)q/M$ for the pinning
frequency generated by the spatial modulation, the conditions are:
\begin{enumerate}\itemsep2pt
\item[(C1)] \emph{Rigid-kink validity.} $\varepsilon_1,\varepsilon_2\ll1$, so
that the induced pinning and driving forces are small compared with the internal
restoring force of the kink. Since the kink width is $\sqrt2$, this also
requires $q,k=\mathcal O(1)$ rather than $\gg1$; for $q\sqrt2\gg1$ the overlap
coefficients \eqref{eq:A1_closed} and \eqref{eq:A2_closed} are exponentially
small and the reduction degenerates.
\item[(C2)] \emph{Weak damping relative to the pinning oscillation.}
$\eta\ll\Omega_0$, so that the unperturbed separatrix survives as an
approximate invariant structure over the timescale of the drive. Equivalently
$\eta\ll\sqrt{\varepsilon_1A_1q/M}$, which couples the admissible damping to
$\varepsilon_1$.
\item[(C3)] \emph{Bias below depinning.} $2|\Gamma|<\varepsilon_1A_1(q)$, so
that the tilted effective potential retains its saddle and the homoclinic
structure is not destroyed outright.
\item[(C4)] \emph{Splitting above the exponentially small floor.}
$\varepsilon_2\gg e^{-\pi\omega/2\Omega_0}$, so that the first-order Melnikov
term dominates the beyond-all-orders corrections; see
Remark~\ref{rem:beyond_all_orders}.
\item[(C5)] \emph{Weak wobbling.} $|a|\ll1$ in \eqref{eq:2CC_ansatz}, which by
\eqref{eq:a_2cc} amounts to
$\varepsilon_2|A_{2,a}(k)|\ll M\,|\Omega_{\rm int}^2-\omega^2+i\eta\omega|$.
\end{enumerate}
The parameter set used in all numerical work below,
$\varepsilon_1=0.10$, $\varepsilon_2=0.15$, $\eta=0.05$, $\Gamma=0.002$,
$q=k=1$, satisfies (C1)--(C5): it gives $\Omega_0=0.191$ so that
$\eta/\Omega_0=0.26$, a bias ratio $2|\Gamma|/(\varepsilon_1A_1)=0.116$, and a
wobbling amplitude below $10^{-1}$ throughout. Condition (C2) is the most
demanding of the five and is satisfied only marginally; the consequences of
this are examined in Sec.~\ref{sec:numerics}.

The driven and spatially modulated $\phi^4$ model
\eqref{eq:pde}--\eqref{eq:mods}, together with the unperturbed
kink solution \eqref{eq:kink} and the scaling assumptions outlined
above, provides the starting point for the collective-coordinate
reductions derived in the next section.

\section{Collective-coordinate reductions}
\label{sec_cc}

In this section, we derive reduced collective-coordinate descriptions for the
driven and spatially modulated $\phi^4$ model introduced in Sec.~\ref{sec:model}.
The goal is to obtain finite-dimensional dynamical systems that capture the
essential mechanisms governing kink motion under spatial pinning,
dissipation, bias, and traveling-wave forcing, while remaining amenable to
a rigorous dynamical-systems analysis.

We proceed in two stages.
First, we derive a translation-only reduction in which the kink is assumed
to remain rigid and is described solely by the time-dependent position of
its center.
Second, we construct a refined two--collective--coordinate (2-CC) reduction
that incorporates the internal (shape) mode of the $\phi^4$ kink and its
leading-order back-reaction on translational motion.
Both reductions are formulated in a constraint-consistent manner and will
serve as the basis for the Melnikov analysis carried out in Sec.~\ref{sec:melnikov}.

\subsection{Translation-only collective-coordinate reduction}
\label{subsec:cc_1cc}

Assuming weak perturbations, we adopt a collective--coordinate approximation
\[
\phi(x,t)\approx \phi_K(x;X(t)).
\]

Projecting \eqref{eq:pde} onto the translational mode, i.e., multiplying by $\phi_{K,x}(x;X)$ and integrating over $x$, we obtain an effective equation of motion for the kink center,
\begin{equation}
M\ddot X + \eta M\dot X
=
-\frac{d}{dX}V(X)
+ \varepsilon_2 A_2\cos(kX-\omega t)
-2\Gamma.
\label{eq:cc}
\end{equation}

The effective mass is given by
\begin{equation}
M=\int_{-\infty}^{\infty} \left(\phi_{K,x}(x;X)\right)^2 dx
= \frac{2\sqrt{2}}{3}.
\label{eq:M_mass_short}
\end{equation}
The spatial modulation $\mathcal{F}(x)$ induces a pinning force through
\[
-\partial_x(\mathcal{F}\phi_x)
= -\phi_{xx}
- \varepsilon_1\sin(qx)\phi_{xx}
- \varepsilon_1 q\cos(qx)\phi_x,
\]
and projecting these additional terms yields, to leading order, a periodic force
\begin{equation*}
F_{\mathcal F}(X)
= -\varepsilon_1 A_1(q) \cos(qX),
\end{equation*}
where the overlap coefficient
\begin{equation}
A_1
=
\int_{-\infty}^{\infty}
\left[
\sin(qx)\phi_{K,xx}
+ q\cos(qx)\phi_{K,x}
\right]\phi_{K,x}(x;X)\,dx
\label{eq:A1_general_def_first}
\end{equation}
is independent of $X$ by translational invariance. Thus the induced effective potential is
\begin{equation*}
V(X)=\frac{\varepsilon_1 A_1(q)}{q}\,\sin(qX),
\qquad
F_{\mathcal F}(X)=-\frac{dV}{dX},
\end{equation*}
whose minima and maxima are located at $qX=\mp\pi/2 \ (\mathrm{mod}\ 2\pi)$; the
saddle points of the associated pendulum dynamics sit at the maxima.
The spatio--temporal modulation $\lambda(t,x)$ produces an effective driving force
\begin{equation*}
F_\lambda(X,t)
=
-\varepsilon_2
\int_{-\infty}^{\infty}
\sin(kx-\omega t)\,
\phi_K(\phi_K^2-1)\,
\phi_{K,x}\,dx
\approx
\varepsilon_2 A_2\cos(kX-\omega t),
\end{equation*}
where
\begin{equation}
A_2
=
\int_{-\infty}^{\infty}
\sin(kx)\,
\phi_K(x)\big(\phi_K(x)^2-1\big)\,
\phi_{K,x}(x)\,dx.
\label{eq:A2_general_def_first}
\end{equation}

\begin{proposition}[Translation-only collective-coordinate reduction]
\label{prop:cc1}
Assume $\varepsilon_1,\varepsilon_2,\eta,\Gamma$ are sufficiently small so that the field remains
close to a rigid $\phi^4$ kink with slowly varying center $X(t)$.
Then projecting \eqref{eq:pde} onto the translational mode yields the reduced equation
\eqref{eq:cc}, with mass \eqref{eq:M_mass_short} and forcing terms defined by
\eqref{eq:A1_general_def_first}--\eqref{eq:A2_general_def_first}.
\end{proposition}

\begin{proof}
With $\phi(x,t)=\phi_K(x;X(t))$ one has $\phi_t=-\dot X\phi_{K,x}$ and
$\phi_{tt}=-\ddot X\phi_{K,x}+\dot X^2\phi_{K,xx}$. Multiplying by
$\phi_{K,x}$ and integrating,
\[
\int \phi_{tt}\,\phi_{K,x}\,dx
=-\ddot X\!\int\!\phi_{K,x}^2\,dx+\dot X^2\!\int\!\phi_{K,xx}\phi_{K,x}\,dx
=-M\ddot X ,
\]
because the second integral equals
$\tfrac12\big[\phi_{K,x}^2\big]_{-\infty}^{+\infty}=0$ by the exponential decay
of $\phi_{K,x}$; equivalently, $\phi_{K,xx}$ is odd and $\phi_{K,x}$ is even in
$\xi$. Similarly $\int\phi_t\phi_{K,x}\,dx=-M\dot X$. The unperturbed part of
$-\partial_x(\mathcal F\phi_x)+\lambda\phi(\phi^2-1)$ vanishes identically on the
kink, so the remaining contributions are exactly the $\mathcal O(\varepsilon_1)$
and $\mathcal O(\varepsilon_2)$ projections \eqref{eq:A1_general_def_first} and
\eqref{eq:A2_general_def_first}, while the bias projects onto
$-\Gamma[\phi_K]_{-\infty}^{+\infty}=-2\Gamma$. Collecting terms and changing
the overall sign gives \eqref{eq:cc}.
\end{proof}

Equation \eqref{eq:cc} has the structure of a driven, damped particle moving in an effective
periodic landscape.
In particular, the $\mathcal{F}$-modulation creates a pinning potential $V(X)$, while the
traveling-wave modulation generates a time-periodic drive with spatial phase.
This combination is precisely the setting in which separatrix dynamics and its perturbation-induced
splitting become central.

\subsection{Closed-form overlap integrals and explicit reduced equation}
\label{subsec:cc_overlaps}

We now provide the computational details leading to explicit overlap integrals and then extend the analysis to a two--collective--coordinate (2-CC) setting by incorporating the $\phi^4$ internal mode. We rewrite the model with the same modulations \eqref{eq:mods} using again the kink \eqref{eq:kink}. Introducing the stretched coordinate $\xi=(x-X)/\sqrt2$ gives
\begin{equation}
\phi_K(\xi)=\tanh\xi,\qquad 
\phi_{K,x}=\frac{1}{\sqrt2}\sech^2\xi,\qquad
\phi_{K,xx}=-\sech^2\xi\,\tanh\xi,\qquad
dx=\sqrt2\,d\xi .
\label{eq:xi_defs}
\end{equation}
The translational mass can be computed explicitly as
\begin{align}
M
&:=\int_{-\infty}^{\infty}(\phi_{K,x})^2\,dx
=\int_{-\infty}^{\infty}\frac{1}{2}\sech^4\xi\,(\sqrt2\,d\xi)
=\frac{\sqrt2}{2}\int_{-\infty}^{\infty}\sech^4\xi\,d\xi
=\frac{\sqrt2}{2}\cdot \frac{4}{3}
=\frac{2\sqrt2}{3}.
\label{eq:mass_detailed}
\end{align}

To compute the spatial-modulation induced forcing, we start from
\begin{equation*}
-\partial_x(\mathcal F\phi_x)= -\phi_{xx}-\varepsilon_1\partial_x\big(\sin(qx)\phi_x\big)
= -\phi_{xx}-\varepsilon_1\Big(\sin(qx)\phi_{xx}+q\cos(qx)\phi_x\Big).
\end{equation*}
The unperturbed $-\phi_{xx}$ is already balanced by $\phi(\phi^2-1)$ for the kink, thus the $\mathcal F$--induced forcing comes from the $\mathcal O(\varepsilon_1)$ part:
\begin{equation}
F_{\mathcal F}(X):=
-\varepsilon_1\int_{-\infty}^{\infty}
\Big(\sin(qx)\phi_{K,xx}+q\cos(qx)\phi_{K,x}\Big)\phi_{K,x}\,dx.
\label{eq:F_F_def}
\end{equation}
Substituting the kink derivatives and changing variables 
$$x=X+\sqrt2\xi,\qquad dx=\sqrt2\,d\xi,\qquad \beta=q\sqrt2,
$$
 yields
\begin{align}
F_{\mathcal F}(X)
&=-\varepsilon_1\int_{-\infty}^{\infty}
\Big(\sin(q(X+\sqrt2\xi))\,(-\sech^2\xi\tanh\xi)
+q\cos(q(X+\sqrt2\xi))\,\frac{1}{\sqrt2}\sech^2\xi\Big)\,
\sech^2\xi\,d\xi\nonumber\\
&=-\varepsilon_1\int_{-\infty}^{\infty}
\Big(-\sin(qX+ \beta\xi)\,\sech^4\xi\tanh\xi
+\frac{q}{\sqrt2}\cos(qX+\beta\xi)\,\sech^4\xi\Big)\,d\xi.
\label{eq:F_F_xi}
\end{align}
Using parity ($\sech^4$ even, $\tanh$ odd), only the $\cos(qX)$-component survives.
It should be stressed that \emph{both} terms in \eqref{eq:F_F_xi} contribute: the
first through its $\cos(qX)\sin(\beta\xi)$ part (even integrand, since
$\sech^4\xi\tanh\xi$ is odd) and the second through its $\cos(qX)\cos(\beta\xi)$
part. Only after the integration by parts recorded below do the two collapse onto
the same Fourier transform of $\sech^4\xi$:
\begin{equation}
F_{\mathcal F}(X)
=-\varepsilon_1\,A_1(q)\,\cos(qX).
\label{eq:F_F_final}
\end{equation}
A compact way to obtain $A_1(q)$ is to note the identity $\tanh\xi\,\sech^4\xi=-\frac14(\sech^4\xi)'$, so that
\[
\int_{-\infty}^{\infty}\sin(\beta\xi)\,\tanh\xi\,\sech^4\xi\,d\xi
=\frac{\beta}{4}\int_{-\infty}^{\infty}\cos(\beta\xi)\,\sech^4\xi\,d\xi ,
\]
and both terms of \eqref{eq:F_F_xi} reduce to the single transform
\[
\int_{-\infty}^{\infty}\cos(\beta\xi)\sech^4\xi\,d\xi
=\frac{\pi\beta(\beta^2+4)}{6\sinh(\pi\beta/2)} .
\]
Collecting the two contributions gives
$A_1(q)=\frac{q}{2\sqrt2}\int_{-\infty}^{\infty}\cos(\beta\xi)\sech^4\xi\,d\xi$;
setting $\beta=q\sqrt2$ then yields the closed form
\begin{equation}
A_1(q)=\frac{\pi q^2(q^2+2)}{6\,\sinh\!\big(\frac{\pi q}{\sqrt2}\big)}.
\label{eq:A1_closed}
\end{equation}
The corresponding pinning potential is therefore
\begin{equation}
V(X)=\frac{\varepsilon_1 A_1(q)}{q}\,\sin(qX).
\label{eq:V_pin_again}
\end{equation}

For the travelling-wave modulation, the extra forcing is
\begin{equation}
F_{\lambda}(X,t):=
-\varepsilon_2\int_{-\infty}^{\infty}
\sin(kx-\omega t)\,\phi_K(\phi_K^2-1)\,\phi_{K,x}\,dx.
\label{eq:F_lambda_def}
\end{equation}
Since $\phi_K(\phi_K^2-1)=\tanh\xi(\tanh^2\xi-1)=-\tanh\xi\,\sech^2\xi$, we obtain after the same change of variables:
\begin{align*}
F_{\lambda}(X,t)
&=-\varepsilon_2\int_{-\infty}^{\infty}
\sin(k(X+\sqrt2\xi)-\omega t)\,(-\tanh\xi\,\sech^2\xi)\,
\sech^2\xi\,d\xi\nonumber\\
&=\varepsilon_2\int_{-\infty}^{\infty}
\sin\big(kX-\omega t+\kappa\xi\big)\,\tanh\xi\,\sech^4\xi\,d\xi,
\qquad \kappa:=k\sqrt2.
\end{align*}
By parity, only the $\cos(kX-\omega t)$ component survives and we find
\begin{equation}
F_{\lambda}(X,t)=\varepsilon_2\,A_2(k)\,\cos(kX-\omega t),
\label{eq:F_lambda_final}
\end{equation}
where
\begin{equation}
A_2(k)=\int_{-\infty}^{\infty}\sin(\kappa\xi)\,\tanh\xi\,\sech^4\xi\,d\xi
=\frac{\pi k^2(k^2+2)}{6\,\sinh\!\big(\frac{\pi k}{\sqrt2}\big)}.
\label{eq:A2_closed}
\end{equation}
The last equality follows from the same integration-by-parts identity $\tanh\xi\,\sech^4\xi=-\frac14(\sech^4\xi)'$, giving
\[
\int_{-\infty}^{\infty} \sin(\kappa\xi)\tanh\xi\sech^4\xi\,d\xi
=\frac{\kappa}{4}\int_{-\infty}^{\infty} \cos(\kappa\xi)\sech^4\xi\,d\xi,
\]
and the known transform of $\sech^4\xi$ stated above.

Collecting \eqref{eq:mass_detailed}, \eqref{eq:F_F_final}, \eqref{eq:F_lambda_final}, and the bias projection
\begin{equation*}
\int_{-\infty}^{\infty}(-\Gamma)\phi_{K,x}\,dx
=-\Gamma\big(\phi_K(+\infty)-\phi_K(-\infty)\big)=-2\Gamma,
\end{equation*}
the translation-only collective-coordinate equation becomes
\begin{equation}
M\ddot X+\eta M\dot X
=
-\varepsilon_1A_1(q)\cos(qX)
+\varepsilon_2A_2(k)\cos(kX-\omega t)
-2\Gamma.
\label{eq:Xeq}
\end{equation}

\begin{proposition}[Closed-form overlap coefficients and explicit translation CC model]
For the kink \eqref{eq:kink} and modulations \eqref{eq:mods}, the overlap coefficients generated by
the spatially modulated gradient term and the traveling-wave nonlinear modulation admit the closed forms
\eqref{eq:A1_closed} and \eqref{eq:A2_closed}. Consequently, the translation-only collective-coordinate
equation takes the explicit form \eqref{eq:Xeq}.
\end{proposition}

\begin{proof}
The derivation follows by substituting \eqref{eq:xi_defs} into the projection formulas
\eqref{eq:F_F_def} and \eqref{eq:F_lambda_def}, using parity and the Fourier transforms of $\sech^4$,
leading to \eqref{eq:A1_closed}--\eqref{eq:A2_closed}, and then collecting terms as in \eqref{eq:Xeq}.
\end{proof}

The explicit evaluation of these overlap integrals yields a reduced
collective-coordinate system with fully analytical coefficients,
which makes the subsequent dynamical-systems analysis particularly
transparent.

Equation \eqref{eq:Xeq} is the explicit translation-only reduced model
used in the remainder of the analysis. It shows clearly how each
physical ingredient enters: a conservative periodic pinning term
($\varepsilon_1 A_1$), a nonautonomous traveling-wave drive
($\varepsilon_2 A_2$), linear damping ($\eta$), and a constant bias
($\Gamma$). This explicit structure is particularly convenient for
formulating Melnikov conditions directly in terms of the control
parameters and will be used in Sec.~\ref{sec:melnikov}.

\subsection{Phase conventions for traveling-wave forcing}
\label{subsec:cc_phase}

To accommodate arbitrary phase conventions, it is convenient to write the traveling-wave forcing as  $\sin(kX-\omega t)$ and $\cos(kX-\omega t)$ without loss of generality, we write the traveling-wave forcing equivalently as
\begin{equation}
F_{\rm drive}(X,t)=\varepsilon_2\Big[
C_c\,\cos(kX-\omega t)+C_s\,\sin(kX-\omega t)
\Big],
\label{eq:drive_sincos}
\end{equation}
with real constants $C_c,C_s$. This two-term form is completely equivalent to a single shifted cosine,
\begin{equation}
F_{\rm drive}(X,t)=\varepsilon_2 R\,\cos\big(kX-\omega t-\varphi\big),
\qquad
R=\sqrt{C_c^2+C_s^2},\quad \tan\varphi=\frac{C_s}{C_c},
\label{eq:drive_phase}
\end{equation}
so one never loses generality by using either $\sin$ or $\cos$ and absorbing phase shifts into $\varphi$.

This representation is convenient in Melnikov computations because the resulting splitting amplitude
is naturally expressed as a single harmonic in the time shift $t_0$, with the phase absorbed into
$\varphi$ (or, equivalently, into a shift of $t_0$).

\subsection{Two--collective--coordinate reduction with internal mode}
\label{subsec:cc_2cc}
We next incorporate the internal (shape) mode and formulate the corresponding 2-CC reduction, while maintaining the weak-wobbling assumption that the back-reaction on the translational dynamics enters at linear order in $a$.
 Linearizing the unperturbed PDE about the kink, $\phi=\phi_K+\psi e^{i\Omega t}$, yields
\begin{equation*}
-\psi_{xx}+\big(3\phi_K^2-1\big)\psi=\Omega^2\psi.
\end{equation*}
In the $\xi$-coordinate this becomes the P\"oschl--Teller problem
\begin{equation*}
-\psi_{\xi\xi}+\big(4-6\sech^2\xi\big)\psi=2\Omega^2\psi,
\end{equation*}
This is the reflectionless P\"oschl--Teller potential with $\ell(\ell+1)=6$, i.e.\
$\ell=2$, whose bound states occur at $2\Omega^2-4=-(\ell-n)^2$, $n=0,1$. It
therefore has the translational mode at $\Omega=0$, with
$\psi\propto\sech^2\xi$, and exactly one internal (shape) bound state with
\begin{equation}
\Omega_{\mathrm{int}}^2=\frac32,
\qquad
\psi_{\mathrm{int}}(\xi)\propto \sech\xi\,\tanh\xi,
\label{eq:int_mode}
\end{equation}
the continuum starting at $\Omega^2=2$. Thus
$\Omega_{\mathrm{int}}=\sqrt{3/2}\simeq 1.2247$ lies strictly below the phonon
edge $\sqrt2\simeq1.4142$, in agreement with the standard $\phi^4$ shape-mode
spectrum \cite{CampbellSchonfeldWingate1983,LizunovaSutcliffe2021}.
A convenient two-mode ansatz is therefore
\begin{equation}
\phi(x,t)\approx \tanh\!\Big(\frac{x-X(t)}{\sqrt2}\Big)
+a(t)\,\psi_{\mathrm{int}}\!\Big(\frac{x-X(t)}{\sqrt2}\Big),
\qquad
\psi_{\mathrm{int}}(\xi)=\sech\xi\,\tanh\xi,
\label{eq:2CC_ansatz}
\end{equation}
where $a(t)$ is assumed small, and orthogonality constraints can be  enforced to avoid double-counting translation (e.g.\ by imposing $\int a\psi_{\mathrm{int}}\phi_{K,x}\,dx=0$ to leading order). Projecting the PDE onto $\phi_{K,x}$ and onto $\psi_{\mathrm{int}}$ gives, at leading order in $(\varepsilon_1,\varepsilon_2,\eta,\Gamma,a)$, the coupled system
\begin{equation}
M\ddot X+\eta M\dot X
=
-\varepsilon_1A_1(q)\cos(qX)
+\varepsilon_2A_2(k)\cos(kX-\omega t)
-2\Gamma
+\alpha(X,t)\,a
+\mathcal K\big[a,\dot a,\dot X,\ddot X\big]
+\mathcal O(a^2),
\label{eq:X_2cc}
\end{equation}
\begin{equation}
m_a\ddot a+\eta m_a\dot a+m_a\Omega_{\mathrm{int}}^2 a
=
\varepsilon_2\,A_{2,a}(k)\,\sin(kX-\omega t)
+\mathcal O(\varepsilon_1 a,a^2),
\label{eq:a_2cc}
\end{equation}
where the internal-mode inertia is
\begin{eqnarray*}
m_a&:=&\int_{-\infty}^{\infty}\psi_{\mathrm{int}}(\xi)^2\,dx
=\sqrt2\int_{-\infty}^{\infty}\sech^2\xi\,\tanh^2\xi\,d\xi
=\sqrt2\int_{-\infty}^{\infty}(\sech^2\xi-\sech^4\xi)\,d\xi\\
&=&\sqrt2\Big(2-\frac{4}{3}\Big)=\frac{2\sqrt2}{3}=M,
\end{eqnarray*}
and the direct overlap of the travelling-wave modulation with the internal mode is
\begin{equation*}
A_{2,a}(k):=
\sqrt2\int_{-\infty}^{\infty}\cos(\kappa\xi)\,\tanh^2\xi\,\sech^3\xi\,d\xi,
\qquad \kappa=k\sqrt2,
\end{equation*}
the factor $\sqrt2$ arising from $dx=\sqrt2\,d\xi$, consistently with the
normalisation used for $m_a$ above.
Using $\tanh^2\xi\,\sech^3\xi=\sech^3\xi-\sech^5\xi$ and the explicit Fourier transforms
\[
\int_{-\infty}^{\infty}\sech^3\xi\cos(\kappa\xi)\,d\xi
=\frac{\pi(\kappa^2+1)}{2\cosh(\pi\kappa/2)},\qquad
\int_{-\infty}^{\infty}\sech^5\xi\cos(\kappa\xi)\,d\xi
=\frac{\pi(\kappa^4+10\kappa^2+9)}{24\cosh(\pi\kappa/2)},
\]
we obtain the closed form
\begin{equation}
A_{2,a}(k)
=\frac{\sqrt2\,\pi\big(-\kappa^4+2\kappa^2+3\big)}{24\,\cosh(\pi\kappa/2)}
=\frac{\sqrt2\,\pi\big(-4k^4+4k^2+3\big)}{24\,\cosh\!\big(\frac{\pi k}{\sqrt2}\big)}.
\label{eq:A2a_closed}
\end{equation}
The phase difference between \eqref{eq:X_2cc} and \eqref{eq:a_2cc} (cosine vs sine) is not accidental: translation couples through an odd kernel, whereas the internal mode has an even forcing kernel.

\begin{remark}[Decoupling wavenumber]
\label{rem:decoupling}
The polynomial factor in \eqref{eq:A2a_closed} vanishes at
$-4k^4+4k^2+3=0$, i.e.\ at
\[
k^\ast=\sqrt{3/2}\simeq 1.2247 .
\]
At this modulation wavenumber the traveling wave does not couple to the shape
mode at all, and $A_{2,a}(k^\ast)=0$: the internal channel of the splitting
analysed in Sec.~\ref{sec:melnikov} closes identically, and $A_{2,a}$ changes
sign across $k^\ast$. Note the numerical coincidence
$k^\ast=\Omega_{\mathrm{int}}$, which is a consequence of the fact that both
quantities are governed by the same $\ell=2$ P\"oschl--Teller structure.
\end{remark}

\begin{proposition}[2-CC reduction and internal-mode forcing overlap]
In the weak-wobbling regime with ansatz \eqref{eq:2CC_ansatz}, projection yields the coupled system
\eqref{eq:X_2cc}--\eqref{eq:a_2cc}. The direct traveling-wave overlap with the internal mode admits
the closed form \eqref{eq:A2a_closed}.
\end{proposition}

\begin{proof}
Insert the ansatz \eqref{eq:2CC_ansatz} into the residual and retain terms of
first order in $a$ and in $(\varepsilon_1,\varepsilon_2,\eta,\Gamma)$.
Projection onto $\phi_{K,x}$ reproduces the computation of Proposition~\ref{prop:cc1} for
the $a$-independent part, and adds the linear-in-$a$ terms discussed in
Sec.~\ref{subsec:cc_constraint}. Projection onto $\psi_{\rm int}$ annihilates
all $a$-independent contributions, since $\psi_{\rm int}$ is odd and the kink
residual is even in $\xi$, and produces
$m_a(\ddot a+\eta\dot a+\Omega_{\rm int}^2 a)$ on the left through
$\mathcal L_K\psi_{\rm int}=\Omega_{\rm int}^2\psi_{\rm int}$ and the
normalisation $m_a=\int\psi_{\rm int}^2\,dx$. On the right, the traveling wave
contributes
\[
\varepsilon_2\!\int\!\sin(kx-\omega t)\,\phi_K(\phi_K^2-1)\,\psi_{\rm int}\,dx ,
\]
whose kernel $\phi_K(\phi_K^2-1)\psi_{\rm int}=-\tanh^2\xi\,\sech^3\xi$ is even, so
that parity leaves only the $\sin(kX-\omega t)$ component, with coefficient
$A_{2,a}(k)$. Writing $\tanh^2\xi\,\sech^3\xi=\sech^3\xi-\sech^5\xi$ and
inserting the two Fourier transforms quoted above gives
\[
A_{2,a}(k)=\sqrt2\,\pi\!\left[\frac{\kappa^2+1}{2}-\frac{\kappa^4+10\kappa^2+9}{24}\right]
\frac{1}{\cosh(\pi\kappa/2)} ,
\]
and simplifying the bracket to $(-\kappa^4+2\kappa^2+3)/24$ yields
\eqref{eq:A2a_closed}.
\end{proof}

A central quantity in \eqref{eq:X_2cc} is the back-reaction coefficient
$\alpha$, which quantifies how internal-mode oscillations bias the translational
dynamics. In a constraint-consistent reduction $\alpha$ is not an adjustable
parameter: it is fixed uniquely once the gauge (constraint) used to define
$(X,a)$ is specified. As we show in Sec.~\ref{subsec:cc_constraint}, it is in
fact strongly constrained: its $\varepsilon$-independent part vanishes
identically, so that $\alpha$ is $\mathcal O(\varepsilon_1,\varepsilon_2)$ and
depends on $X$ and $t$ through the modulations. The symbol $\mathcal K$ in
\eqref{eq:X_2cc} collects the kinematic couplings
$\propto a\ddot X,\ \dot a\dot X,\ \eta a\dot X$ generated by the moving frame,
which are likewise of higher order.

\subsection{Constraint-consistent parametrization and back-reaction coefficient \texorpdfstring{$\alpha$}{alpha}}
\label{subsec:cc_constraint}

To fix the parametrization $(X(t),a(t))$ uniquely (avoid double counting of translations by the
internal-mode coordinate), we impose the standard orthogonality (gauge) constraint
\begin{equation}
\mathcal C(X,a):=\int_{-\infty}^{\infty}\Big(\phi(x,t)-\phi_K(x;X(t))\Big)\,\phi_{K,x}(x;X(t))\,dx=0.
\label{eq:orth_constraint}
\end{equation}
With the two-mode ansatz $\phi=\phi_K(\xi)+a(t)\psi_{\rm int}(\xi)$, $\xi=(x-X)/\sqrt2$,
this constraint is satisfied at leading order because $\psi_{\rm int}$ is odd and $\phi_{K,x}$ is even;
however, when retaining consistently the next-order kinematic terms generated by the moving frame
(e.g.\ contributions proportional to $a\dot X$, $a\ddot X$), the constraint produces additional
terms in the projected $X$-equation. These are the ``constraint-induced'' contributions.

The coefficient $\alpha$ is defined by the linear projection of the $\mathcal O(a)$ residual generated by the two-mode ansatz onto the translational mode. Writing the PDE residual as
\[
\mathcal R[\phi]=\phi_{tt}+\eta\phi_t-\partial_x(\mathcal F\phi_x)+\lambda\,\phi(\phi^2-1)+\Gamma,
\]
expanding $\mathcal R[\phi_K+a\psi_{\mathrm{int}}]$ to first order in $a$, and enforcing the chosen collective-coordinate constraint, we define
\begin{equation}
\alpha:=
-\frac{1}{M}\int_{-\infty}^{\infty}
\Bigg(
\left.\frac{\partial \mathcal R}{\partial a}\right|_{a=0}
\Bigg)\,\phi_{K,x}\,dx .
\label{eq:alpha_proj_general}
\end{equation}

It is instructive to isolate first the $\varepsilon$-independent part of
\eqref{eq:alpha_proj_general}, i.e.\ the contribution of the unmodulated
linearized operator
$\mathcal L_K=-\partial_x^2+\big(3\phi_K^2-1\big)$. One might expect this to
provide the dominant back-reaction; in fact it does not contribute at all.

\begin{lemma}[Vanishing of the leading back-reaction]
\label{lem:alpha_vanishes}
The $\varepsilon$-independent, linear-in-$a$ projection of the residual onto the
translational mode vanishes identically:
\begin{equation}
\int_{-\infty}^{\infty}\big(\mathcal L_K\psi_{\mathrm{int}}\big)\,\phi_{K,x}\,dx=0 .
\label{eq:alpha_proj_operator}
\end{equation}
\end{lemma}

\begin{proof}
Two independent arguments may be given. \emph{(i) Parity.} In the stretched
variable, $\psi_{\mathrm{int}}(\xi)=\sech\xi\,\tanh\xi$ is odd while
$\phi_{K,x}\propto\sech^2\xi$ is even; since $\mathcal L_K$ has even
coefficients it preserves parity, so the integrand is odd and the integral
vanishes. \emph{(ii) Self-adjointness.} $\mathcal L_K$ is self-adjoint on
$L^2(\mathbb R)$ and $\phi_{K,x}$ is its translational zero mode,
$\mathcal L_K\phi_{K,x}=0$; hence
$\langle\mathcal L_K\psi_{\mathrm{int}},\phi_{K,x}\rangle
=\langle\psi_{\mathrm{int}},\mathcal L_K\phi_{K,x}\rangle=0$.
The second argument shows that the result is not an accident of the explicit
$\phi^4$ profile, but a structural consequence of translational invariance of
the unperturbed problem.
\end{proof}

Consequently the two-mode reduction does \emph{not} generate a back-reaction
term $+\alpha a$ with a constant coefficient of order unity. What survives comes
entirely from the modulations and from the kinematics of the moving frame. In
the stretched variable $\xi=(x-X)/\sqrt2$, the linear-in-$a$ residual produced by
the modulated gradient term is obtained from
$-\partial_x\!\left(\mathcal F\,\partial_x(a\psi_{\rm int})\right)$, while the
modulated nonlinearity contributes through the kernel
\[
(3\phi_K^2-1)\psi_{\rm int}(\xi)
=\big(2-3\sech^2\xi\big)\,\sech\xi\,\tanh\xi .
\]
Carrying out the projection and using $\psi_{\rm int}$ odd,
$\psi_{\rm int}'$ even, $\psi_{\rm int}''$ odd, one finds that in each case only
the cosine component of the modulation phase survives, so that
\begin{equation}
\alpha(X,t)=\varepsilon_1\,\alpha_1(q)\cos(qX)
\;+\;\varepsilon_2\,\alpha_2(k)\cos(kX-\omega t),
\label{eq:alpha_xi_general}
\end{equation}
with the explicit overlap coefficients
\begin{align}
\alpha_1(q)&=\frac{1}{M}\int_{-\infty}^{\infty}
\Big[\tfrac12\sin(q\sqrt2\,\xi)\,\psi_{\rm int}''(\xi)
+\tfrac{q}{\sqrt2}\cos(q\sqrt2\,\xi)\,\psi_{\rm int}'(\xi)\Big]\sech^2\xi\,d\xi ,
\label{eq:alpha1_def}\\[2pt]
\alpha_2(k)&=-\frac{1}{M}\int_{-\infty}^{\infty}
\sin(k\sqrt2\,\xi)\,\big(2-3\sech^2\xi\big)\,\psi_{\rm int}(\xi)\,\sech^2\xi\,d\xi .
\label{eq:alpha2_def}
\end{align}
These are absolutely convergent integrals of $\sech$--$\tanh$ kernels and are
readily evaluated to machine precision; representative values are collected in
Table~\ref{tab:alpha}.

\begin{table}[t]
\centering
\begin{tabular}{lrrrrrr}
\hline\noalign{\smallskip}
$q,\,k$ & $0.3$ & $0.5$ & $0.7$ & $1.0$ & $1.5$ & $2.0$\\
\noalign{\smallskip}\hline\noalign{\smallskip}
$\alpha_1(q)$ & $-0.01853$ & $-0.01750$ & $-0.00088$ & $0.05051$ & $0.16179$ & $0.23276$\\
$\alpha_2(k)$ & $-0.14078$ & $-0.20126$ & $-0.22140$ & $-0.17680$ & $0.01156$ & $0.18288$\\
\noalign{\smallskip}\hline
\end{tabular}
\caption{Back-reaction overlap coefficients \eqref{eq:alpha1_def}--\eqref{eq:alpha2_def},
computed with $M=2\sqrt2/3$.}
\label{tab:alpha}
\end{table}

The kinematic couplings collected in $\mathcal K$ arise from the time
derivatives of the moving-frame ansatz. Their common overlap is
\begin{equation*}
\int_{-\infty}^{\infty}\psi_{\rm int}'(\xi)\,\sech^2\xi\,d\xi
=\int_{-\infty}^{\infty}\big(2\sech^5\xi-\sech^3\xi\big)\,d\xi
=\frac{\pi}{4},
\end{equation*}
which is nonzero because $\psi_{\rm int}'$ is even; the resulting terms are
proportional to $a\ddot X$, $\dot a\dot X$ and $\eta a\dot X$ and are therefore
bilinear in small quantities. The precise split between
\eqref{eq:alpha_xi_general} and $\mathcal K$ depends on the gauge, but their sum
does not.

\begin{remark}[Ordering of the internal channel]
\label{rem:ordering}
Since the internal amplitude is itself driven at $\mathcal O(\varepsilon_2)$ by
\eqref{eq:a_2cc}, Lemma~\ref{lem:alpha_vanishes} together with
\eqref{eq:alpha_xi_general} implies that the back-reaction of the shape mode on
the translational equation is of order $\varepsilon_1\varepsilon_2$ and
$\varepsilon_2^2$ --- one order beyond the direct traveling-wave forcing. The
shape mode therefore cannot bias the kink centre at leading order, and it enters
the splitting analysis of Sec.~\ref{sec:melnikov} only through this
second-order, modulation-mediated coupling. This is a structural statement: it
requires no gauge fixing beyond \eqref{eq:orth_constraint} and involves no
adjustable parameter.
\end{remark}

With the reduced 2-CC system in place, we next address separatrix splitting in a manner that does not
rely on informal energy arguments alone.
Because the reduced 2-CC system evolves in a higher-dimensional phase space, it is natural to employ the
adjoint variational equation formulation of Melnikov theory, which provides a clean and invariant
definition of the splitting function. The existence of separatrix splitting and the associated chaotic dynamics is a geometric property of the reduced phase space and does not depend on the specific choice of collective-coordinate parametrization.

\section{Separatrix splitting and chaotic dynamics}
\label{sec:melnikov}

The translation-only reduction provides the simplest setting in which a
homoclinic separatrix arises in the reduced phase space, making it a
natural starting point for the Melnikov analysis. Melnikov theory provides a classical analytical criterion for
transverse intersections of stable and unstable manifolds in
perturbed Hamiltonian systems. The method has been widely applied
to a variety of nonlinear dynamical systems; see for example
\cite{Dua2015,Li2016,LiWeiMoroz2025,DuanWei2024}, where the method has been
extended to hybrid piecewise-smooth systems under impulsive, noisy and
quasiperiodic excitation. 

The two--collective--
coordinate system then extends this framework by incorporating the
internal mode and its coupling to translational motion.

\subsection{Unperturbed separatrix and perturbative setting}
\label{subsec:sep_unpert}

In the Hamiltonian limit $\eta=\varepsilon_2=\Gamma=0$, Eq.~\eqref{eq:cc}
describes a particle moving in a periodic potential and therefore possesses a
separatrix. Since the pinning force is $-\varepsilon_1A_1(q)\cos(qX)$, the
natural pendulum angle is the \emph{shifted} variable
\begin{equation*}
u=qX+\frac{\pi}{2},
\end{equation*}
in terms of which the unperturbed dynamics reads
\begin{equation*}
\ddot u + \Omega_0^2\sin u = 0,
\qquad
\Omega_0^2=\frac{\varepsilon_1 A_1 q}{M}.
\end{equation*}
A separatrix solution is
\begin{equation*}
u_0(t)=4\arctan\!\left(e^{\Omega_0 t}\right),
\qquad
\dot u_0(t)=\frac{2\Omega_0}{\cosh(\Omega_0 t)}.
\end{equation*}
Treating damping, bias, and the traveling-wave drive as small perturbations,
Melnikov theory can be applied. The Melnikov function associated with the
separatrix $X_0(t)=\big(u_0(t)-\pi/2\big)/q$ is
\begin{equation*}
\mathcal M(t_0)
=
\int_{-\infty}^{\infty}
\dot X_0(t)
\Big[
-\eta M\dot X_0(t)
+\varepsilon_2 A_2\cos(kX_0(t)-\omega(t+t_0))
-2\Gamma
\Big]\, dt.
\end{equation*}
The individual contributions can be evaluated explicitly. The damping term yields
\begin{equation*}
\int_{-\infty}^{\infty}\dot X_0^2\, dt
=
\frac{1}{q^2}\int_{-\infty}^{\infty}\dot u_0^2\, dt
=
\frac{8\Omega_0}{q^2},
\end{equation*}
while the bias term gives
\begin{equation*}
\int_{-\infty}^{\infty}\dot X_0\, dt
=
X_0(+\infty)-X_0(-\infty)
=
\frac{2\pi}{q}.
\end{equation*}
The driving contribution is oscillatory in $t_0$ and can be written as
\[
\int_{-\infty}^{\infty}
\dot X_0\cos(kX_0-\omega(t+t_0))\,dt
=
B\sin(\omega t_0+\delta),
\]
with amplitude $B$ and phase $\delta$ determined by $(\omega,k,\Omega_0)$.
Hence the Melnikov function takes the form
\begin{equation*}
\mathcal M(t_0)
=
-\eta M\frac{8\Omega_0}{q^2}
-\frac{4\pi\Gamma}{q}
+\varepsilon_2 A_2 B\sin(\omega t_0+\delta),
\end{equation*}
and the existence of simple zeros of $\mathcal M(t_0)$ implies transverse
intersections of the stable and unstable manifolds of the perturbed system,
leading to Smale horseshoes and chaotic dynamics in the kink-center motion. A
sufficient condition for chaos is therefore
\begin{equation}
\abs{\varepsilon_2 A_2 B}
>
\eta M\frac{8\Omega_0}{q^2}
+\frac{4\pi\abs{\Gamma}}{q}.
\label{eq:chaos_cond_first}
\end{equation}
This establishes the basic translation-only mechanism by which the interplay
between spatial pinning induced by $F(x)$ and the traveling-wave modulation of
$\lambda(t,x)$ leads to separatrix splitting and chaotic kink dynamics.

We next make the above reduction fully explicit by evaluating the
overlap integrals in closed form. This yields analytical coefficients
for the reduced dynamics and provides a concrete baseline for the
Melnikov analysis developed in the next subsection. In particular,
it allows the splitting amplitude to be expressed directly in terms
of the physical parameters of the model before extending the analysis
to the 2--CC system.

\subsection{Melnikov analysis for the translation-only reduction}
\label{subsec:melnikov_1cc}

Throughout the Melnikov analysis we assume that the perturbation
parameters $\varepsilon_2$, $\eta$, and $\Gamma$ are sufficiently
small so that the dynamics can be treated as a weak perturbation
of the unperturbed Hamiltonian separatrix. 

In the Hamiltonian limit of \eqref{eq:Xeq}
($\eta=\varepsilon_2=\Gamma=0$), the kink center obeys
\begin{equation*}
M\ddot X=-\varepsilon_1A_1(q)\cos(qX),
\end{equation*}
and letting $u=qX+\pi/2$ gives the standard pendulum equation
\begin{equation}
\ddot u+\Omega_0^2\sin u=0,
\qquad
\Omega_0^2=\frac{\varepsilon_1A_1(q)\,q}{M}.
\label{eq:Omega0_def}
\end{equation}
A convenient separatrix is
\begin{equation}
u_0(t)=4\arctan(e^{\Omega_0 t}),\qquad
\dot u_0(t)=\frac{2\Omega_0}{\cosh(\Omega_0 t)},\qquad
X_0(t)=\frac{u_0(t)-\pi/2}{q}.
\label{eq:sep}
\end{equation}
Since the pinning potential is $2\pi$-periodic in $u$, the natural phase space
is the cylinder $\mathbb S^1\times\mathbb R$, on which the saddles at
$u=\pi\ (\mathrm{mod}\ 2\pi)$ are a single equilibrium and \eqref{eq:sep} is a
\emph{homoclinic} orbit. This is also the phase space in which the Poincar\'e
sections of Sec.~\ref{sec:numerics} are computed. We use the term homoclinic
throughout; on the universal cover the same orbit connects two distinct saddles
and would be called heteroclinic, but nothing in the analysis depends on this
distinction.
Treating $(\eta,\varepsilon_2,\Gamma)$ as small perturbations, the Melnikov
function can be written in the form
\begin{equation}
\mathcal M(t_0)=
\int_{-\infty}^{\infty}\dot X_0(t)\Big[
-\eta M\dot X_0(t)-2\Gamma+F_{\rm drive}(X_0(t),t+t_0)
\Big]dt,
\label{eq:Melnikov_general}
\end{equation}
with the explicit damping and bias contributions
\begin{equation}
\int_{-\infty}^{\infty}\dot X_0^2\,dt
=\frac{1}{q^2}\int_{-\infty}^{\infty}\frac{4\Omega_0^2}{\cosh^2(\Omega_0 t)}\,dt
=\frac{8\Omega_0}{q^2},
\qquad
\int_{-\infty}^{\infty}\dot X_0\,dt=\frac{2\pi}{q}.
\label{eq:Melnikov_const}
\end{equation}
giving
\begin{equation}
\mathcal M(t_0)=
-\eta M\frac{8\Omega_0}{q^2}
-\frac{4\pi\Gamma}{q}
+\varepsilon_2\,\mathcal M_{\rm drive}(t_0),
\label{eq:Melnikov_split}
\end{equation}
where
\begin{equation*}
\mathcal M_{\rm drive}(t_0)=
\int_{-\infty}^{\infty}\dot X_0(t)\Big[
C_c\cos\big(kX_0(t)-\omega(t+t_0)\big)+
C_s\sin\big(kX_0(t)-\omega(t+t_0)\big)
\Big]dt.
\end{equation*}

\begin{proposition}[Melnikov function for the translation-only separatrix]
In the perturbed regime of \eqref{eq:Xeq}, the Melnikov function associated
with the unperturbed separatrix \eqref{eq:sep} is given by
\eqref{eq:Melnikov_general}--\eqref{eq:Melnikov_split}, with the damping and
bias contributions \eqref{eq:Melnikov_const}. If $\mathcal M(t_0)$ has a
simple zero, the stable and unstable manifolds intersect transversely and the
reduced dynamics exhibits Smale horseshoes.
\end{proposition}

\begin{proof}
Write the reduced system as a planar perturbed Hamiltonian system,
$\dot z=f(z)+\varepsilon g(z,t)$ with $z=(X,\dot X)$,
$f$ the pendulum vector field of \eqref{eq:sep} and
$\varepsilon g=(0,\ -\eta\dot X-2\Gamma/M+F_{\rm drive}/M)$, the perturbation
being $2\pi/\omega$-periodic in $t$. For such systems the Melnikov function is
the integral along the unperturbed orbit of the wedge product
$f\wedge g$, which here reduces to $M\!\int\dot X_0\,g_2\,dt$, i.e.\ to
\eqref{eq:Melnikov_general}. The damping and bias integrals
\eqref{eq:Melnikov_const} follow from $\dot u_0=2\Omega_0\sech(\Omega_0t)$ and
$\int\sech^2=2$, and from $u_0(+\infty)-u_0(-\infty)=2\pi$. Since $\mathcal M$
measures, to first order in $\varepsilon$, the distance between the stable and
unstable manifolds of the perturbed saddle, a simple zero implies a transverse
intersection; the Smale--Birkhoff theorem then gives an invariant set on which
the Poincar\'e map is topologically conjugate to a shift, i.e.\ a horseshoe
\cite{Wiggins1996}.
\end{proof}

\noindent
Transverse intersections of the stable and unstable manifolds of the
perturbed saddle imply the existence of a Smale horseshoe and therefore
chaotic dynamics in the reduced phase space; see, for example,
\cite{Wiggins1996}.

Because the forcing is $2\pi/\omega$-periodic in time and the separatrix orbit
is fixed up to the time shift $t_0$, the Melnikov function inherits a single
$\omega$-harmonic dependence on $t_0$. This is the key mechanism by which a
nonautonomous perturbation produces transversal manifold intersections in a
separatrix-bearing system.

In the commensurate case $k=q$, a closed form for the drive amplitude can be
obtained using explicit trigonometric identities along the separatrix. Writing
$u=qX+\pi/2$ and using \eqref{eq:sep} one has
\begin{equation*}
\sin u_0(t)=-2\,\sech(\Omega_0 t)\tanh(\Omega_0 t),
\qquad
\cos u_0(t)=1-2\,\sech^2(\Omega_0 t),
\qquad
\dot u_0(t)=2\Omega_0\sech(\Omega_0 t).
\end{equation*}
Since $qX_0(t)=u_0(t)-\pi/2$, the phase $kX_0$ differs from $u_0$ only by a
constant shift and can be absorbed into $(C_c,C_s)$ (equivalently into
$\varphi$ in \eqref{eq:drive_phase}). Therefore it is enough to compute one
canonical case, e.g.\ $F_{\rm drive}=\cos(qX-\omega t)=\sin(u_0-\omega t)$, and
then reintroduce the general phase shift via \eqref{eq:drive_phase}. One finds
that only the $\sin(\omega t_0)$ harmonic survives, and that it receives
\emph{two} distinct contributions. Writing $\alpha_\omega:=\omega/\Omega_0$ and
$s=\Omega_0 t$,
\begin{align*}
\int_{-\infty}^{\infty}\dot u_0\,\cos u_0\,\cos(\omega t)\,dt
&=2\int_{-\infty}^{\infty}\sech s\,\big(1-2\sech^2 s\big)\cos(\alpha_\omega s)\,ds
=-\frac{2\pi\alpha_\omega^2}{\cosh(\pi\alpha_\omega/2)},\\
\int_{-\infty}^{\infty}\dot u_0\,\sin u_0\,\sin(\omega t)\,dt
&=-4\int_{-\infty}^{\infty}\sech^2 s\,\tanh s\,\sin(\alpha_\omega s)\,ds
=-\frac{2\pi\alpha_\omega^2}{\sinh(\pi\alpha_\omega/2)} .
\end{align*}
The first uses the transforms
\[
\int_{-\infty}^{\infty}\sech(\tau)\cos(\alpha\tau)\,d\tau
=\frac{\pi}{\cosh(\pi\alpha/2)},
\qquad
\int_{-\infty}^{\infty}\sech^3(\tau)\cos(\alpha\tau)\,d\tau
=
\frac{\pi(\alpha^2+1)}{2\cosh(\pi\alpha/2)},
\]
while the second follows from
$\tanh s\,\sech^2 s=-\tfrac12(\sech^2 s)'$ together with
\begin{equation*}
\int_{-\infty}^{\infty}\sech^2(\tau)\cos(\alpha\tau)\,d\tau
=\frac{\pi\alpha}{\sinh(\pi\alpha/2)} .
\end{equation*}
Both channels contribute at the same order and must be retained. The result is
\begin{equation}
\mathcal M_{\rm drive}(t_0)
=
B_{q}(\omega)\,\sin(\omega t_0+\delta),
\qquad
B_{q}(\omega)=\frac{2\pi}{q}\Big(\frac{\omega}{\Omega_0}\Big)^{2}
\left[
\frac{1}{\cosh\!\Big(\frac{\pi\omega}{2\Omega_0}\Big)}
+\frac{1}{\sinh\!\Big(\frac{\pi\omega}{2\Omega_0}\Big)}
\right],
\label{eq:B_closed_k_equals_q}
\end{equation}
which we have verified against direct numerical evaluation of
$\mathcal M_{\rm drive}$ to machine precision over several decades in
$\omega/\Omega_0$. Two limits are worth recording: as
$\omega/\Omega_0\to\infty$ both terms behave as
$2e^{-\pi\omega/2\Omega_0}$, so the splitting is exponentially small but twice
as large as the $\cosh$ term alone would suggest; as $\omega/\Omega_0\to0$ one
has $\csch(\pi\alpha_\omega/2)\sim 2/(\pi\alpha_\omega)$ and hence
$B_q=\mathcal O(\omega)$ rather than $\mathcal O(\omega^2)$, which shifts the
low-frequency maximum of the splitting amplitude.
The phase $\delta$ depends only on the chosen $\sin$/$\cos$ convention
and can be absorbed into $\varphi$ in \eqref{eq:drive_phase}. If one keeps the
general forcing \eqref{eq:drive_sincos}, then \eqref{eq:B_closed_k_equals_q}
still holds with the replacement $B_q(\omega)\mapsto R\,B_q(\omega)$. Thus, in
the commensurate case $k=q$, the Melnikov function becomes
\begin{equation*}
\mathcal M(t_0)=
-\eta M\frac{8\Omega_0}{q^2}
-\frac{4\pi\Gamma}{q}
+\varepsilon_2\,R\,B_q(\omega)\,\sin(\omega t_0+\tilde\delta),
\end{equation*}
and a sufficient condition for transverse manifold intersections (hence Smale
horseshoes and chaos in the reduced kink-center dynamics) is
\begin{equation}
\varepsilon_2\,R\,B_q(\omega) >
\eta M\frac{8\Omega_0}{q^2}+\frac{4\pi|\Gamma|}{q}.
\label{eq:chaos_threshold_k_eq_q}
\end{equation}

\begin{corollary}[Sufficient chaos threshold]
A sufficient condition for transverse intersections is the inequality
\eqref{eq:chaos_cond_first} or, in the commensurate case $k=q$,
\eqref{eq:chaos_threshold_k_eq_q}.
\end{corollary}

For $k\neq q$ there is no simple closed form in general because the phase
$kX_0(t)$ is not the same pendulum angle. Nevertheless, the drive contribution
remains a single-harmonic function of $t_0$,
\begin{equation*}
\mathcal M_{\rm drive}(t_0)=B(\omega,k)\,\sin(\omega t_0+\delta),
\end{equation*}
with amplitude determined by the two real integrals
\begin{align*}
B(\omega,k)\cos\delta &=
\int_{-\infty}^{\infty}\dot X_0(t)\,
\sin\big(kX_0(t)-\omega t\big)\,dt,\\
B(\omega,k)\sin\delta &=
\int_{-\infty}^{\infty}\dot X_0(t)\,
\cos\big(kX_0(t)-\omega t\big)\,dt,
\end{align*}
which can be evaluated numerically with high accuracy because $\dot X_0(t)$
decays exponentially as $|t|\to\infty$.

We now move beyond translation-only dynamics by allowing for excitation of the
kink's internal mode. This is essential when the external drive can transfer
energy not only to the center-of-mass motion but also to localized shape
oscillations, which may then feed back into the translational equation.

\subsection{Extended Melnikov formulation for the two--collective--coordinate system}
\label{subsec:extended_melnikov}
\label{subsec:melnikov_2cc}

At this stage, the 2-CC Melnikov analysis can be formulated systematically by
writing the coupled system in standard perturbation form and constructing the
Melnikov function using the adjoint variational equation. Introducing the state
vector $\bm z=(X,v,a,b)^{\mathsf T}$ with $v=\dot X$ and $b=\dot a$, we write
the dynamics as
\begin{equation*}
\dot{\bm z}= \bm f(\bm z)+\varepsilon\,\bm g(\bm z,t),
\end{equation*}
where $\varepsilon$ is a bookkeeping small parameter collecting all weak
effects (dissipation, bias, and traveling-wave forcing) and $\bm f$ denotes the
autonomous unperturbed vector field corresponding to the Hamiltonian pendulum
in $(X,v)$ together with the linear oscillator in $(a,b)$. The unperturbed
homoclinic orbit in the extended phase space is
\[
\bm z_0(t) = \big(X_0(t),\dot X_0(t),0,0\big)^{\mathsf T},
\]
with $X_0$ given by \eqref{eq:sep}. Let $D\bm f(\bm z_0(t))$ be the Jacobian
matrix of $\bm f$ evaluated along $\bm z_0(t)$. The adjoint variational
equation is
\begin{equation}
\dot{\bm \psi}(t) = -\big(D\bm f(\bm z_0(t))\big)^{\mathsf T}\bm \psi(t),
\label{eq:adjoint}
\end{equation}
and we choose a bounded adjoint solution $\bm \psi(t)$ spanning the orthogonal
complement of the tangent direction $\dot{\bm z}_0(t)$. The Melnikov function
is then defined by the standard pairing \cite{Wiggins1996,yam92}
\begin{equation}
\mathcal{M}(t_0)=\int_{-\infty}^{\infty}
\bm \psi(t)\cdot \bm g\big(\bm z_0(t),t+t_0\big)\,dt,
\label{eq:Melnikov_adjoint}
\end{equation}
and the existence of a simple zero of $\mathcal{M}(t_0)$ implies transverse
intersections of stable and unstable manifolds (Smale horseshoes) for the
perturbed system in the extended phase space.

In the weak-wobbling regime, the dominant transverse splitting is captured by a
reduction of \eqref{eq:Melnikov_adjoint} to a scalar integral along the
translational separatrix, with the internal mode entering through its bounded
forced response along $X_0(t)$. We write the internal response as the unique
bounded solution $a_0(t;t_0)$ of the forced internal-mode equation along the
separatrix,
\begin{equation}
M\ddot a_0+\eta M\dot a_0+M\Omega_{\mathrm{int}}^2 a_0
=
\varepsilon_2 A_{2,a}(k)\sin\big(kX_0(t)-\omega(t+t_0)\big),
\label{eq:a_forced_on_sep}
\end{equation}
by which we mean the unique solution \emph{bounded on the whole line}. It
should be stressed that $a_0$ does \emph{not} decay as $t\to\pm\infty$: since
$X_0(t)$ tends to the finite limits $X_0(\pm\infty)$, the right-hand side of
\eqref{eq:a_forced_on_sep} tends to a persistent sinusoid of frequency $\omega$,
and the bounded response accordingly approaches the corresponding periodic
steady state. Existence and uniqueness are nevertheless immediate. With
$\eta>0$ the homogeneous solutions behave as $e^{-\eta t/2}e^{\pm i\nu t}$,
$\nu=(\Omega_{\mathrm{int}}^2-\eta^2/4)^{1/2}$, so two bounded solutions would
differ by a homogeneous solution that is unbounded as $t\to-\infty$ unless
identically zero; and existence follows from the convolution representation
\begin{equation}
a_0(t;t_0)=\frac{\varepsilon_2A_{2,a}(k)}{M}\int_{-\infty}^{t}
G(t-s)\,\sin\big(kX_0(s)-\omega(s+t_0)\big)\,ds,
\qquad
G(\tau)=\frac{e^{-\eta\tau/2}\sin(\nu\tau)}{\nu},
\label{eq:green}
\end{equation}
whose kernel is absolutely integrable. The Melnikov integral below converges
regardless, because it carries the exponentially decaying weight $\dot X_0(t)$.
Thus the adjoint-based Melnikov function admits, at leading order, a scalar
Melnikov representation of the form
\begin{equation}
\mathcal{M}(t_0)
=
\int_{-\infty}^{\infty}\dot X_0(t)\Big[
-\eta M\dot X_0(t)-2\Gamma+\varepsilon_2 A_2(k)\cos\big(kX_0(t)-\omega(t+t_0)\big)
+\alpha\big(X_0(t),t\big)\,a_0(t;t_0)
\Big]dt,
\label{eq:Melnikov_2cc_effective_Oa}
\end{equation}
where the last term is now explicit and no longer hidden in $\mathcal O(a_0)$,
and $\alpha(X,t)$ is the modulation-induced back-reaction coefficient
\eqref{eq:alpha_xi_general} established in Sec.~\ref{subsec:cc_constraint}.

The dependence of \eqref{eq:Melnikov_2cc_effective_Oa} on the time shift $t_0$
can now be determined \emph{exactly}, without any adiabatic assumption. Since
\eqref{eq:a_forced_on_sep} is linear and $t_0$ enters only through the phase
factor $e^{-i\omega t_0}$, writing the forcing as
$\mathrm{Im}\big[e^{i(kX_0(t)-\omega t)}e^{-i\omega t_0}\big]$ and using
linearity gives
\begin{equation}
a_0(t;t_0)=\frac{\varepsilon_2A_{2,a}(k)}{M}\,
\mathrm{Im}\!\left[\mathcal A(t)\,e^{-i\omega t_0}\right],
\qquad
\mathcal A(t)=\int_{-\infty}^{t}G(t-s)\,e^{i(kX_0(s)-\omega s)}\,ds .
\label{eq:A_complex}
\end{equation}
Substituting into \eqref{eq:Melnikov_2cc_effective_Oa} and retaining the leading
back-reaction $\varepsilon_1\alpha_1(q)\cos(qX)$ from
\eqref{eq:alpha_xi_general}, the internal contribution is
\begin{equation}
\mathcal M_{\rm int}(t_0)
=\frac{\varepsilon_1\varepsilon_2\,\alpha_1(q)A_{2,a}(k)}{M}\,
\mathrm{Im}\!\left[\mathcal C_{\rm int}\,e^{-i\omega t_0}\right],
\qquad
\mathcal C_{\rm int}=\int_{-\infty}^{\infty}\dot X_0(t)\cos\big(qX_0(t)\big)\,\mathcal A(t)\,dt .
\label{eq:Cint}
\end{equation}
The single-harmonic dependence on $t_0$ is therefore \emph{exact}, not an
approximation; what requires care is only the evaluation of the amplitude
$\mathcal C_{\rm int}$. Both $\mathcal A$ and $\mathcal C_{\rm int}$ are
computed here by numerical solution of the linear equation
\eqref{eq:a_forced_on_sep} along the separatrix, to a relative accuracy of
$10^{-9}$.

It is instructive to compare this with the adiabatic estimate obtained by
treating the forcing as a pure $\omega$-harmonic, which would give
$\mathcal A\sim\big[(\Omega_{\rm int}^2-\omega^2)-i\eta\omega\big]^{-1}$ and
hence the familiar resonant denominator. That estimate is unreliable here, and
in a way worth recording: the resonant response builds up only where the
forcing has settled onto a single frequency, i.e.\ as $t\to\pm\infty$, whereas
the Melnikov integral \eqref{eq:Cint} is weighted by $\dot X_0(t)$ and is
therefore concentrated near $t=0$, precisely where the instantaneous forcing
frequency $\omega-k\dot X_0(t)$ departs most from $\omega$. The resonance is
thus filtered out of the splitting integral. Numerically, the adiabatic formula
misestimates $|\mathcal C_{\rm int}|$ by factors ranging from $2.4$ at
$\omega=0.19$ to $0.06$ at $\omega=1.6$, and at $\omega=\Omega_{\rm int}$ it
overestimates it by a factor $1.45$. The exact amplitude varies smoothly through
the internal-mode frequency, with no sharp peak.

Writing also the translational contribution in complex form,
\begin{equation}
\mathcal M_{\rm tr}(t_0)=\varepsilon_2 A_2(k)\,
\mathrm{Re}\!\left[\mathcal Z_{\rm tr}\,e^{-i\omega t_0}\right],
\qquad
\mathcal Z_{\rm tr}=\int_{-\infty}^{\infty}\dot X_0(t)\,e^{i(kX_0(t)-\omega t)}\,dt ,
\label{eq:Ztr}
\end{equation}
the full Melnikov function becomes
\begin{equation}
\mathcal{M}(t_0)=
-\eta M\frac{8\Omega_0}{q^2}
-\frac{4\pi\Gamma}{q}
+\varepsilon_2\,\mathcal R(\omega,k)\,\cos\big(\omega t_0-\delta\big),
\label{eq:Melnikov_two_contrib}
\end{equation}
where the two channels combine \emph{as complex amplitudes},
\begin{equation}
\mathcal R(\omega,k)=\Big|\,A_2(k)\,\mathcal Z_{\rm tr}
+\tfrac{\varepsilon_1}{M}\alpha_1(q)A_{2,a}(k)\,\mathcal C_{\rm int}\Big|,
\qquad
\delta=\arg\Big(A_2\mathcal Z_{\rm tr}
+\tfrac{\varepsilon_1}{M}\alpha_1A_{2,a}\mathcal C_{\rm int}\Big).
\label{eq:R_combined}
\end{equation}
Writing $\mathcal B_{\rm tr}=|A_2\mathcal Z_{\rm tr}|$ and
$\mathcal B_{\rm int}=|\varepsilon_1\alpha_1A_{2,a}\mathcal C_{\rm int}/M|$ for
the individual magnitudes, the triangle inequality gives
\begin{equation}
\mathcal R(\omega,k)\ \le\ \mathcal B_{\rm tr}+\mathcal B_{\rm int},
\label{eq:triangle}
\end{equation}
with equality only if the two channels are in phase. The scalar sum is
therefore an \emph{upper} bound on the splitting amplitude and must not be used
inside a sufficient condition for chaos. For the parameters of
Sec.~\ref{sec:numerics} the two contributions turn out to be close to
quadrature, with relative phase between $-148^\circ$ and $+90^\circ$ over the
range scanned, so that $\mathcal B_{\rm tr}+\mathcal B_{\rm int}$ exceeds
$\mathcal R$ by between $0.05\%$ at $\omega=0.19$ and $6\%$ at $\omega=1.6$.
The correction is numerically small here, but the criterion below is stated in
terms of $\mathcal R$, which is the correct quantity in general.

A sufficient condition for transverse manifold intersections, and hence for
horseshoes in the reduced 2-CC dynamics, is that $\mathcal M(t_0)$ possess a
simple zero. From \eqref{eq:Melnikov_two_contrib} this reads
\begin{equation}
\varepsilon_2\,\mathcal R(\omega,k)
>
\eta M\frac{8\Omega_0}{q^2}+\frac{4\pi|\Gamma|}{q},
\label{eq:chaos_criterion_2cc}
\end{equation}
with $\mathcal R$ the phase-correct combination \eqref{eq:R_combined}. By
\eqref{eq:triangle} the more convenient expression
$\mathcal B_{\rm tr}+\mathcal B_{\rm int}$ may be used only as an upper bound on
$\mathcal R$, and hence only to establish that the criterion \emph{fails}, never
that it holds.

Two structural features deserve emphasis. First, the explicit factor
$\varepsilon_1$ in $\mathcal B_{\rm int}$, inherited from
Lemma~\ref{lem:alpha_vanishes}: the internal contribution to the splitting is
not merely numerically small, it is \emph{one order higher} in the perturbation
than the translational one, and it is switched on only by the spatial
modulation. In a homogeneous medium ($\varepsilon_1=0$) the shape mode would not
contribute to the splitting at this order at all. Second, by
Remark~\ref{rem:decoupling} the internal channel closes identically at
$k=k^\ast=\sqrt{3/2}$, irrespective of $\omega$, since $A_{2,a}(k^\ast)=0$.
Both statements are testable and neither involves a free parameter.

Taken together with the filtering argument following \eqref{eq:Cint}, these
observations lead to a conclusion that is more definite, if less dramatic, than
a resonance-driven lowering of the threshold: within the present reduction the
internal mode is \emph{doubly} suppressed in the splitting function, once by the
structural $\varepsilon_1$ prefactor and once by the exponential separatrix
weighting, which removes the resonant denominator from
$\mathcal C_{\rm int}$. The shape mode is nevertheless genuinely and resonantly
excited along the trajectory; what does not happen is that this excitation feeds
back appreciably onto the translational splitting. Section~\ref{sec:numerics}
tests both halves of this statement directly.

\begin{proposition}[Adjoint Melnikov formulation and the internal-mode contribution]
For the extended state $\bm z=(X,v,a,b)^{\mathsf T}$, the Melnikov function
defined via the adjoint variational equation \eqref{eq:adjoint} and pairing
\eqref{eq:Melnikov_adjoint} admits the effective representation
\eqref{eq:Melnikov_2cc_effective_Oa} in the weak-wobbling regime, where the
internal-mode response along the separatrix is the unique bounded solution of
\eqref{eq:a_forced_on_sep}. The dependence on $t_0$ is a single harmonic
exactly, the two channels combining as the complex amplitudes
\eqref{eq:R_combined}, and the internal contribution carries the prefactor
$\varepsilon_1\alpha_1(q)$ of Lemma~\ref{lem:alpha_vanishes}.
\end{proposition}

\begin{proof}
\textit{(Sketch.)}
For $\eta>0$ the $(a,b)$ block of $\bm f$ is a linear oscillator with strictly
negative real parts of its eigenvalues, i.e.\ a stable focus, while the $(X,v)$
block carries the hyperbolic saddle and its separatrix. The bounded adjoint
solution of \eqref{eq:adjoint} along $\bm z_0(t)$ therefore has a component in
the $(X,v)$ directions proportional to $(-\ddot X_0,\dot X_0)$ --- the standard
planar Melnikov vector --- while its $(a,b)$ components decay exponentially in
both time directions. Pairing with $\bm g$ as in \eqref{eq:Melnikov_adjoint} and
using that the internal coordinate is slaved to the translational separatrix
through the unique bounded solution of \eqref{eq:a_forced_on_sep} then yields
\eqref{eq:Melnikov_2cc_effective_Oa}. The neglected terms are of higher order in
the weak-wobbling ordering of Remark~\ref{rem:ordering}.
\end{proof}

\begin{corollary}[Practical 2-CC chaos criterion]
A sufficient condition for transverse manifold intersections in the reduced
2-CC dynamics is \eqref{eq:chaos_criterion_2cc}.
\end{corollary}

The decomposition \eqref{eq:Melnikov_two_contrib} separates two conceptually
distinct mechanisms for separatrix splitting: a direct, translation-mediated
contribution already present in the 1-CC model, and an indirect contribution
mediated by the internal mode and its back-reaction. The analysis above shows
that the second is subordinate to the first throughout: it carries the
structural factor $\varepsilon_1$ of Lemma~\ref{lem:alpha_vanishes}, and the
resonant enhancement one might expect near $\omega\simeq\Omega_{\rm int}$ does
not materialise in $\mathcal C_{\rm int}$ because of the separatrix weighting.
The value of the decomposition therefore lies less in exhibiting a second route
to chaos --- there is none at this order --- than in making precise why there is
not.

\subsection{Frequency dependence of Melnikov thresholds}
\label{subsec:omega_scan}

To expose the frequency dependence inherent in the adjoint Melnikov splitting
\eqref{eq:Melnikov_two_contrib}, we compute an $\omega$--scan of the splitting
amplitudes and the associated Melnikov forcing amplitudes. Crucially, this scan
evaluates the Melnikov integrals \emph{along the unperturbed separatrix} and
therefore constitutes a parameter-to-splitting map, rather than a
time-integration study of the reduced dynamics.

Throughout this subsection we use the parameter set
\begin{equation}
\varepsilon_1=0.10,\qquad \eta=0.05,\qquad \Gamma=0.002,\qquad q=k=1,
\label{eq:scan_params}
\end{equation}
for which the closed forms derived above give
$M=0.942809$, $A_1=A_2=0.344769$, $A_{2,a}=0.119060$, $\alpha_1(1)=0.050513$,
and hence
\begin{equation*}
\Omega_0=\sqrt{\frac{\varepsilon_1A_1q}{M}}=0.191228,
\qquad
\Omega_{\rm int}=\sqrt{3/2}=1.224745,
\end{equation*}
with a dissipative-plus-bias right-hand side
$\eta M\,8\Omega_0/q^2+4\pi|\Gamma|/q=9.725\times10^{-2}$. The scan is carried
out over $0<\omega\le 2$, so that it passes \emph{through} the internal-mode
resonance rather than stopping short of it.

Figure~\ref{fig:melnikov_pair}(a) reports the Melnikov forcing amplitude
$\varepsilon_2^{\mathrm{M}}(\omega)$ predicted by the translation-only and
2-CC criteria. The threshold attains its minimum
$\varepsilon_2^{\mathrm{M}}\simeq0.051$ at $\omega\simeq0.24$, comfortably
inside the perturbative regime, and grows rapidly thereafter: this reflects the
exponentially small separatrix--forcing overlap
$\sim e^{-\pi\omega/2\Omega_0}$ contained in \eqref{eq:B_closed_k_equals_q}. The
two thresholds remain close over the whole range, indicating that the dominant
contribution to separatrix splitting is the translational channel in
\eqref{eq:Melnikov_two_contrib}.

This is made explicit in Fig.~\ref{fig:melnikov_pair}(b), which compares the
translation-mediated amplitude $\mathcal B_{\mathrm{tr}}$ with the
internal-mode-mediated contribution $\mathcal B_{\mathrm{int}}$ across the same
$\omega$--scan, now on a logarithmic scale so that both are visible. Both are
computed from the exact amplitudes \eqref{eq:Cint} and \eqref{eq:Ztr}. At the
operating frequency of Sec.~\ref{sec:numerics} the ratio is
$\mathcal B_{\mathrm{int}}/\mathcal B_{\mathrm{tr}}=5.3\times10^{-4}$; at
$\omega=\Omega_{\rm int}$ it rises only to $1.8\times10^{-2}$, and
$\mathcal B_{\rm int}$ varies smoothly through the internal-mode frequency
without a sharp peak. This is the filtering effect discussed after
\eqref{eq:Cint}: the resonant response of the shape mode develops where the
forcing has settled onto a single frequency, whereas the splitting integral is
concentrated near $t=0$, where it has not. Had the adiabatic estimate been used
instead, the same ratio would have been reported as
$3.0\times10^{-2}$ at resonance, overestimating it by a factor $1.45$.

Figure~\ref{fig:melnikov_pair}(b) also shows a sharp narrow dip in
$\mathcal B_{\mathrm{int}}$ near $\omega\simeq0.31$. This is not a numerical
artefact. The amplitude $\mathcal C_{\rm int}$ is complex, and its real part
changes sign at $\omega=0.31086$ while its imaginary part remains of order
$10^{-3}$, so that $|\mathcal C_{\rm int}|$ falls to $1.7\times10^{-3}$ ---
smaller by a factor $8.8\times10^{2}$ than its value at the operating frequency
of Sec.~\ref{sec:numerics}. The internal channel therefore very nearly closes at
this driving frequency, in the same way that it closes identically at the
decoupling wavenumber $k^\ast$ of Remark~\ref{rem:decoupling}, but here through
a cancellation between the forced response accumulated before and after the
separatrix passage rather than through a vanishing overlap coefficient.

The smallness of $\mathcal B_{\mathrm{int}}$ therefore rests on two independent
and equally structural mechanisms: the $\varepsilon_1$ prefactor of
Lemma~\ref{lem:alpha_vanishes}, and the separatrix weighting that suppresses the
resonant denominator.

\emph{Why no internal-mode chaos?}
Within the present weak-wobbling two--collective--coordinate reduction, the
internal (shape) mode behaves as a linearly driven and damped oscillator
slaved to the translational separatrix $X_0(t)$. As such, it does not possess
an independent homoclinic structure and therefore cannot
generate Melnikov-type chaos on its own. Its role is necessarily indirect: it
enters the Melnikov function only through the second-order back-reaction on the
\emph{translational} separatrix splitting, becoming most effective when
$\omega\simeq\Omega_{\mathrm{int}}$. The smallness of $\mathcal B_{\mathrm{int}}$
in Fig.~\ref{fig:melnikov_pair}(b) is therefore a consequence of the reduced
phase-space structure --- ``no internal-mode \emph{separatrix}'' --- rather than
an indication of suppressed internal-mode dynamics. Indeed, the shape mode is
strongly and resonantly excited near $\omega\simeq\Omega_{\mathrm{int}}$, as the
direct simulations of Sec.~\ref{sec:numerics} confirm; what the analysis above
shows is that this excitation does not feed back appreciably onto the
translational splitting, and hence does not lower the chaos threshold at this
order.

\begin{remark}[Range of validity of the classical Melnikov estimate]
\label{rem:beyond_all_orders}
Since $\Omega_0=\mathcal O(\sqrt{\varepsilon_1})$, the splitting amplitude
\eqref{eq:B_closed_k_equals_q} carries the factor
$\exp(-\pi\omega/2\Omega_0)=\exp\big(-c\,\omega/\sqrt{\varepsilon_1}\big)$ and is
therefore beyond all orders in $\varepsilon_1$ in the rapid-forcing regime
$\omega\gg\Omega_0$. In that regime the first-order Melnikov estimate is not
automatically justified, and exponential-asymptotics techniques are required to
control the remainder
\cite{HolmesMarsdenScheurle1988,DelshamsSeara1997}. Consistency of the perturbative ordering also demands
$\varepsilon_2\gg e^{-c/\sqrt{\varepsilon_1}}$, so that the drive-induced
splitting dominates the neglected corrections. Both requirements are met for the
parameter set \eqref{eq:scan_params} in the range where
$\varepsilon_2^{\mathrm{M}}$ is smallest, and the scan of
Fig.~\ref{fig:melnikov_pair} should be read with this restriction in mind: for
$\omega$ well beyond $\Omega_{\rm int}$ the curves are best regarded as
indicating the trend of the threshold rather than its precise value.
\end{remark}

\begin{figure}[t]
\centering
\begin{minipage}[t]{0.48\linewidth}
\centering
\includegraphics[width=\linewidth]{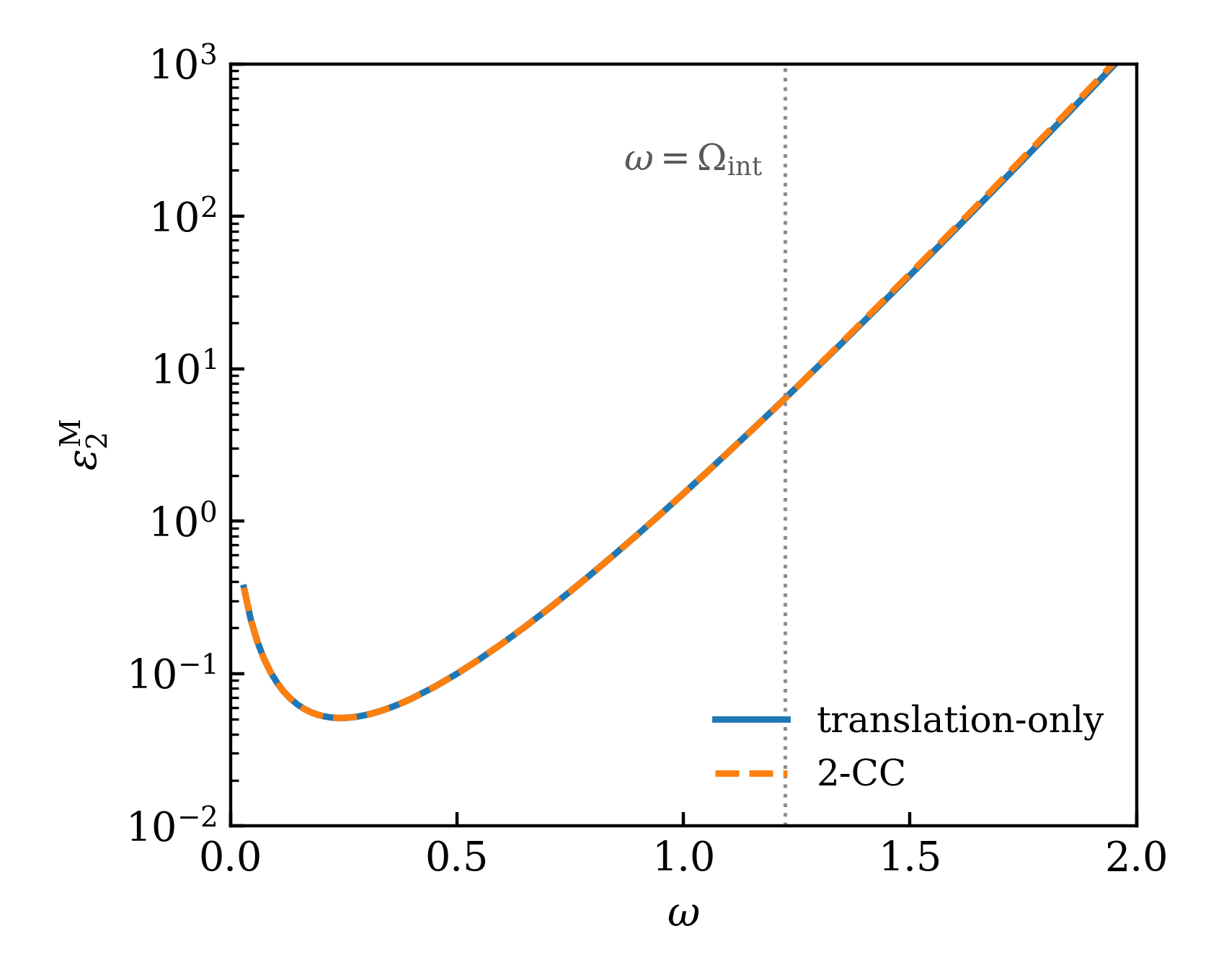}

\vspace{0.5em}
\small (a) Melnikov forcing amplitude $\varepsilon_2^{\mathrm{M}}(\omega)$ above
which transverse intersections are guaranteed, for the translation-only and 2-CC
criteria and the parameters \eqref{eq:scan_params}. The minimum
$\varepsilon_2^{\mathrm{M}}\simeq0.051$ at
$\omega\simeq0.24$ lies well inside the perturbative regime; the subsequent rapid
growth reflects the exponentially small overlap between the separatrix velocity
and high-frequency forcing. The dotted line marks
$\omega=\Omega_{\mathrm{int}}=\sqrt{3/2}$.
\end{minipage}\hfill
\begin{minipage}[t]{0.48\linewidth}
\centering
\includegraphics[width=\linewidth]{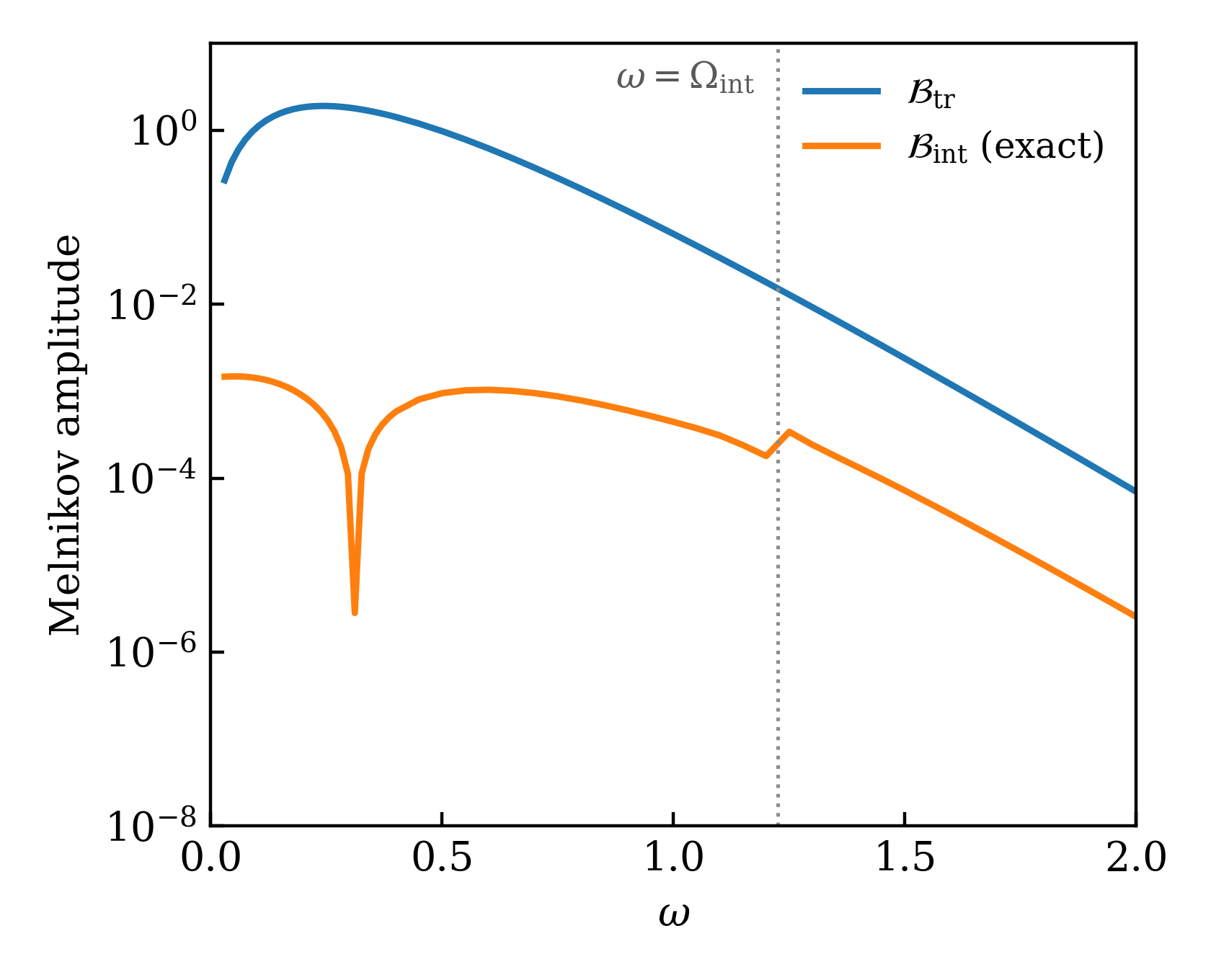}

\vspace{0.5em}
\small (b) Translation-mediated ($\mathcal B_{\mathrm{tr}}$) and
internal-mode-mediated ($\mathcal B_{\mathrm{int}}$) Melnikov amplitudes versus
$\omega$, on a logarithmic scale, both computed from the exact amplitudes
\eqref{eq:Cint} and \eqref{eq:Ztr}. The translational contribution dominates
throughout. $\mathcal B_{\mathrm{int}}$ carries the structural $\varepsilon_1$
suppression of Lemma~\ref{lem:alpha_vanishes} and passes smoothly through
$\omega=\Omega_{\mathrm{int}}$ without a sharp peak, the resonant denominator
being filtered out by the separatrix weighting in \eqref{eq:Cint}.
\end{minipage}
\caption{Frequency dependence of the Melnikov thresholds and splitting amplitudes
for the parameter set \eqref{eq:scan_params}. Panel (a) shows the critical forcing
threshold predicted by the translation-only and 2-CC criteria; panel (b) compares
the corresponding translation-mediated and internal-mode-mediated splitting
contributions. Together the two panels show that, in the parameter regime
considered, the translational channel dominates the separatrix splitting
mechanism, while the internal channel is resonantly enhanced near
$\omega=\Omega_{\mathrm{int}}$ without becoming dominant.}
\label{fig:melnikov_pair}
\end{figure}

\section{Numerical simulations}
\label{sec:numerics}

In this section, we present numerical simulations that complement the analytical
results obtained from the collective-coordinate reductions and the Melnikov
analysis developed above. Since Melnikov theory applies rigorously to the reduced
finite-dimensional dynamical system, rather than directly to the full field
equation, our numerical investigation focuses on the reduced collective-coordinate
(CC) dynamics. This choice allows for a direct and unambiguous validation of the
analytical predictions concerning separatrix splitting and the emergence of
chaotic motion.

We consider the translation-only collective-coordinate model derived in the
preceding sections, which describes the evolution of the kink center in the
presence of spatial pinning, damping, bias, and traveling-wave forcing. Upon
introducing the shifted variable $u=qX+\pi/2$, the reduced dynamics takes the form of
a periodically forced, weakly dissipative pendulum. In the absence of forcing,
damping, and bias, the system possesses a homoclinic separatrix associated with
the pinning potential induced by the spatial modulation. The Melnikov analysis
predicts that, when the forcing amplitude exceeds a critical threshold depending
on the system parameters, the stable and unstable manifolds of this separatrix
intersect transversely, giving rise to Smale horseshoes and chaotic dynamics.

The reduced system is integrated numerically using a fixed-step fourth-order
Runge--Kutta scheme. A fixed time step is employed in order to preserve the phase
accuracy required for stroboscopic sampling. Initial conditions are chosen in the
vicinity of the unperturbed separatrix, ensuring that the trajectory explores the
region of phase space where separatrix splitting is expected to occur. After
discarding a sufficiently long transient interval, the long-time dynamics is
analyzed by constructing stroboscopic Poincar\'e sections.

For reproducibility we record the parameter values used in all simulations
reported in this section:
\begin{equation}
\varepsilon_1=0.10,\quad
\varepsilon_2=0.15,\quad
\eta=0.05,\quad
\Gamma=0.002,\quad
q=k=1,\quad
\omega=0.19,
\label{eq:num_params}
\end{equation}
integrated with fixed time step $\Delta t=T/400=0.08267$, where
$T=2\pi/\omega=33.069$ is the forcing period, over a total time
$T_{\rm tot}=1.984\times10^{5}$ (6000 forcing periods), of which the first
$3.307\times10^{4}$ (1000 periods) are discarded as transient. With these
values the analytical criterion \eqref{eq:chaos_threshold_k_eq_q} gives
$\varepsilon_2^{\mathrm{M}}=0.0541$, so that the simulations reported below
lie a factor $2.8$ above the predicted onset of transverse intersections. The
corresponding unperturbed pinning frequency is $\Omega_0=0.19123$, so that the
drive is close to the $1{:}1$ resonance $\omega\simeq\Omega_0$.

These values were not selected merely to produce chaos. They were fixed by the
following sequence, in which only the last step involves any numerical search.
The wavenumbers are set to $q=k=1$, the commensurate case, both because it is
the only one admitting the closed form \eqref{eq:B_closed_k_equals_q} and
because the product $A_1(q)\,q$ is near its maximum there, so that a given
$\varepsilon_1$ buys the deepest available pinning potential. Then
$\varepsilon_1=0.10$ is the largest value comfortably satisfying (C1) of
Sec.~\ref{sec:model}; it fixes $\Omega_0=0.191$. The damping $\eta=0.05$ gives
$\eta/\Omega_0=0.26$, which is as large as (C2) permits while still leaving the
separatrix a meaningful structure --- and it must not be made much smaller,
since the dissipative term is what sets the right-hand side of the Melnikov
criterion and hence makes the threshold finite at all. The bias
$\Gamma=0.002$ is an order of magnitude below the depinning value of (C3),
large enough to produce the directed transport visible in
Fig.~\ref{fig:diagnostics_pair}(a) but too small to destroy the saddle. The
driving frequency is then read off Fig.~\ref{fig:melnikov_pair}(a):
$\varepsilon_2^{\mathrm M}(\omega)$ has a well-defined minimum near
$\omega\simeq0.24$, and any $\omega$ in the flat region around it minimizes
the forcing required to cross the Melnikov threshold. Only at this final stage
did we scan: over $\omega\in[0.10,0.22]$ and $\varepsilon_2\in[0.10,0.45]$ we
selected $(\omega,\varepsilon_2)=(0.19,0.15)$ and verified that it yields
$\lambda_{\max}>0$ robustly, for $n_{\rm sub}\in\{300,500\}$ and for three
different initial conditions. The need for that last scan is not an artifact of
the procedure but a property of the system, and is explained by
Fig.~\ref{fig:eps2_sweep}: above the Melnikov threshold chaotic bands alternate
with periodic windows, so no analytical prescription can locate a chaotic
operating point exactly.

The Poincar\'e section is obtained by sampling the reduced trajectory once every
forcing period $T=2\pi/\omega$. The sampled points $(u(t_n),v(t_n))$, with
$t_n=t_0+nT$, are plotted modulo $2\pi$ in the $u$ direction. This construction
yields a two-dimensional discrete map that captures the essential geometry of the
nonautonomous reduced system and provides a natural framework for detecting
chaotic invariant sets predicted by Melnikov theory.

A representative Poincar\'e section of the reduced collective-coordinate dynamics
is shown in Fig.~\ref{fig:poincare_pair}(a). The parameters are chosen within the
Melnikov regime, with the forcing amplitude exceeding the analytical chaos
threshold and the forcing frequency comparable to the characteristic pinning
frequency. The resulting phase-space portrait exhibits an extended chaotic layer,
characterized by a cloud of points with no smooth invariant curves, indicating
the breakdown of regular motion and the onset of chaotic dynamics.

To clarify the mechanism responsible for the observed chaos, the unperturbed
separatrix of the Hamiltonian pinning problem is overlaid on the Poincar\'e
section in Fig.~\ref{fig:poincare_pair}(b). In the plotted variables the
separatrix is the lens $v=\pm2\Omega_0\sin(u/2)$, with the saddle points at the
left and right edges of the frame. The chaotic set is seen to straddle this
curve: it occupies regions both inside the lens, where the unperturbed motion
librates, and outside it, where the unperturbed motion rotates, and it crosses
the curve repeatedly.

This is the decisive geometrical observation. In the unperturbed system the
separatrix is an invariant curve and therefore an impenetrable barrier between
libration and rotation: no orbit can pass from one region to the other. A
single invariant set spanning both regions can exist only once the stable and
unstable manifolds of the perturbed saddle intersect transversely, which is
precisely what the Melnikov condition \eqref{eq:chaos_threshold_k_eq_q}
predicts for $\varepsilon_2>\varepsilon_2^{\mathrm{M}}$. The Poincar\'e
section thus provides direct phase-space evidence that the barrier has been
destroyed by the traveling-wave forcing.

We stress what the figure does \emph{not} show. At the forcing amplitude used
here, $\varepsilon_2/\varepsilon_2^{\mathrm{M}}=2.8$, the homoclinic tangle
is no longer a thin layer hugging the unperturbed separatrix; it has thickened
into an extended attractor with visible folded, fractal-like structure. A thin
stochastic layer of the classical kind would require driving much closer to
threshold, where the attractor is correspondingly harder to resolve
numerically. The separatrix should therefore be read as the organizing
skeleton of the chaotic set rather than as its envelope.

\begin{figure}[t]
\centering
\begin{minipage}[t]{0.48\linewidth}
\centering
\includegraphics[width=\linewidth]{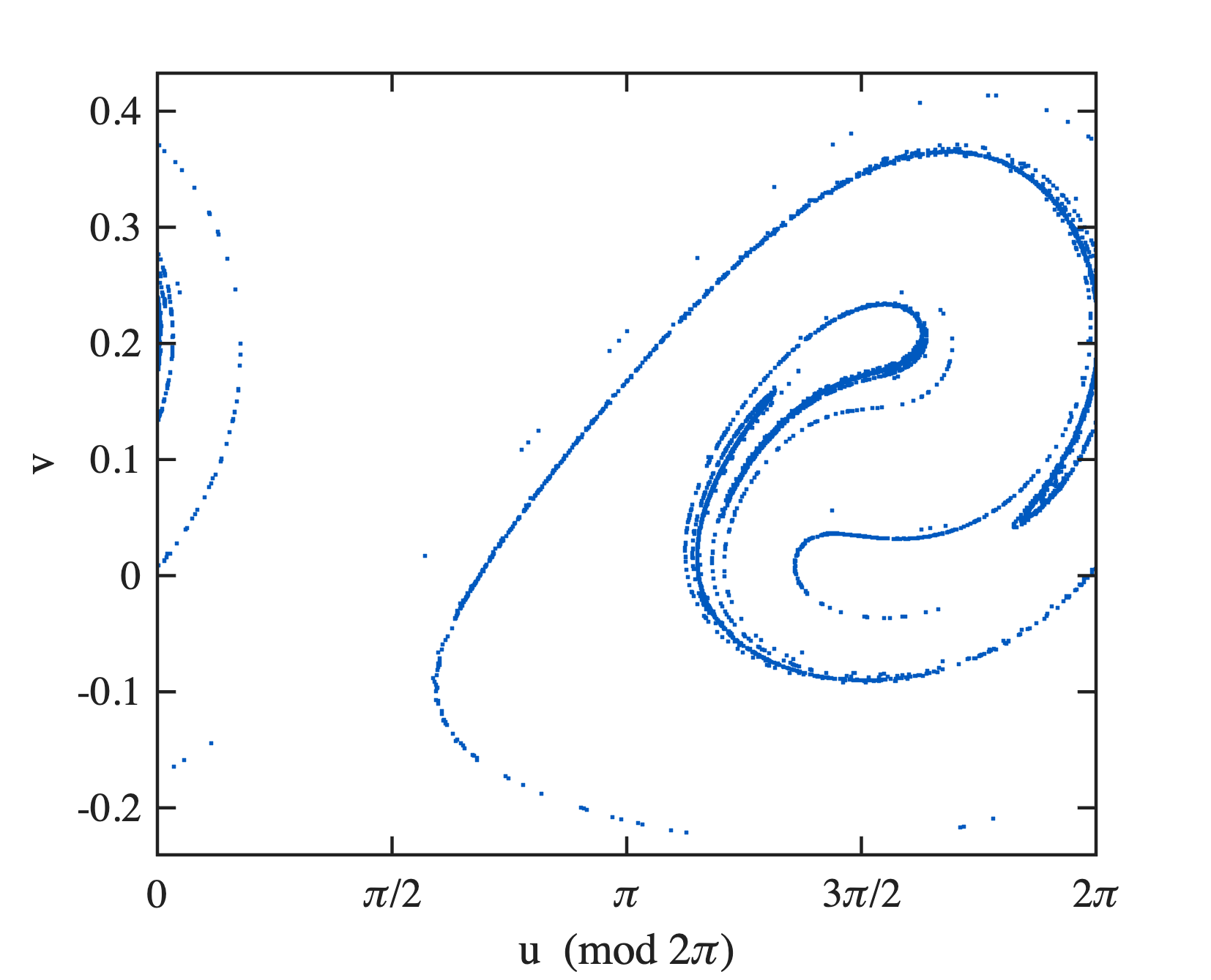}

\vspace{0.5em}
\small (a) Stroboscopic Poincar\'e section of the reduced translation-only
collective-coordinate system. The variables $(u,v)$ are sampled once every
forcing period and plotted modulo $2\pi$ in $u$. The extended cloud of points
indicates the breakdown of invariant curves and the presence of chaotic
dynamics in the reduced system.
\end{minipage}\hfill
\begin{minipage}[t]{0.48\linewidth}
\centering
\includegraphics[width=\linewidth]{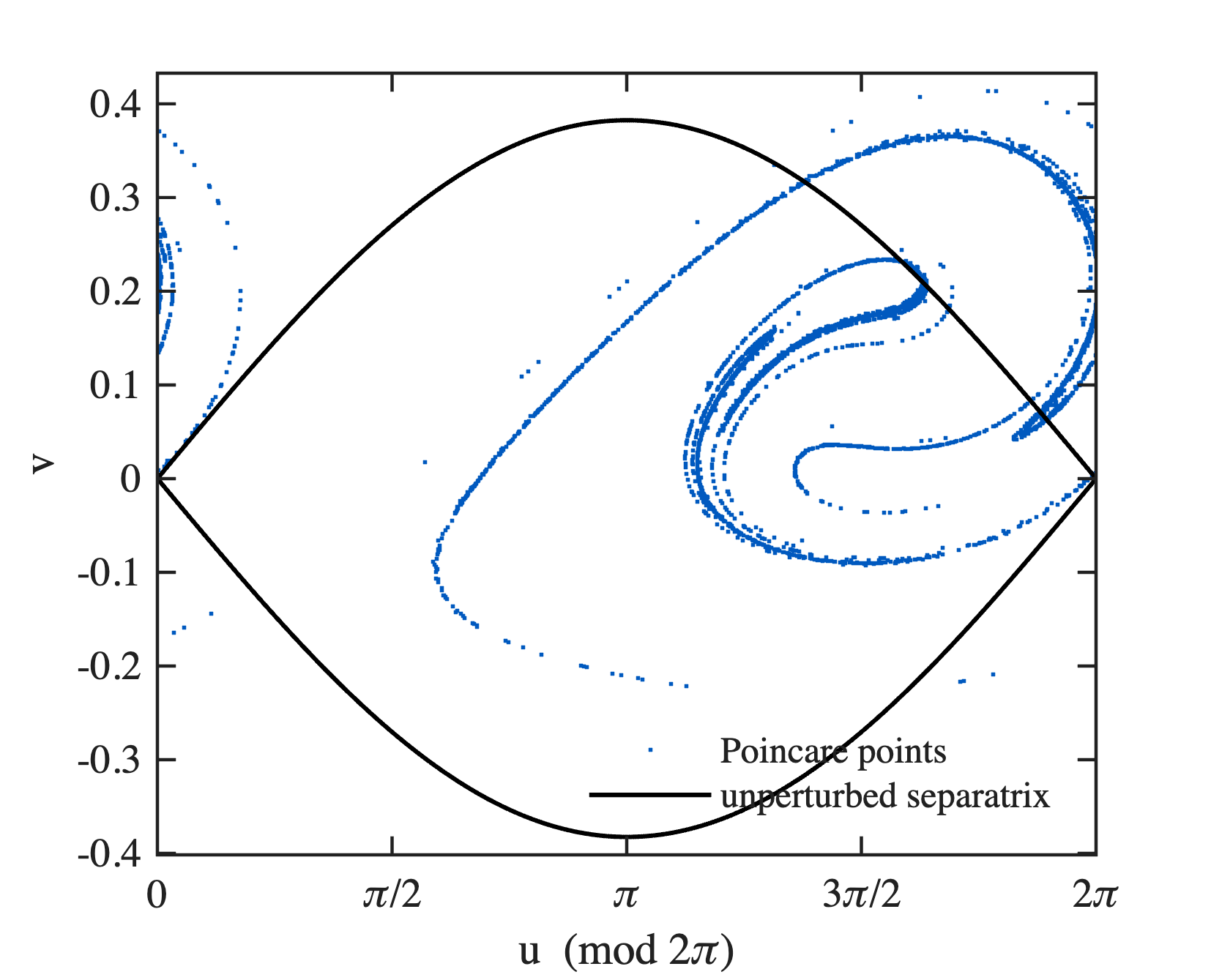}

\vspace{0.5em}
\small (b) Stroboscopic Poincar\'e section of the reduced collective-coordinate
system with the unperturbed separatrix $v=\pm2\Omega_0\sin(u/2)$ overlaid; the
saddle points sit at the left and right edges of the frame. The chaotic set
straddles the separatrix, occupying both the librational interior and the
rotational exterior. Since the unperturbed separatrix is an invariant barrier
between these two regions, a single invariant set spanning both can exist only
after transverse intersection of the stable and unstable manifolds --- the
phase-space signature of the splitting predicted by Melnikov theory.
\end{minipage}
\caption{Poincar\'e diagnostics for the reduced collective-coordinate dynamics
at the parameters \eqref{eq:num_params}.
(a) Stroboscopic phase portrait of the translation-only reduced system.
(b) The same Poincar\'e section with the unperturbed separatrix overlaid.
The chaotic set is organized about the separatrix and crosses it repeatedly,
consistent with the Melnikov prediction of transverse separatrix splitting and
horseshoe chaos.}
\label{fig:poincare_pair}
\end{figure}

Additional insight into the nature of the reduced dynamics is obtained from the
time evolution of the collective-coordinate variable $u(t)$, shown in
Fig.~\ref{fig:diagnostics_pair}(a). The motion exhibits a net drift
$\langle\dot u\rangle\simeq0.116$, corresponding to a depinned running state,
while displaying irregular fluctuations characteristic of chaotic dynamics.
This behavior demonstrates that chaos in the reduced system can coexist with
directed transport.

The chaotic nature of the reduced dynamics is confirmed quantitatively by
computing the largest Lyapunov exponent, obtained by integrating the tangent
flow alongside the trajectory and renormalizing once per forcing period.
Figure~\ref{fig:diagnostics_pair}(b) shows the convergence of
$\lambda_{\max}$ over long integration times. After an initial transient the
estimate settles onto the strictly positive value
\begin{equation*}
\lambda_{\max}=0.0205,
\end{equation*}
stable to three significant figures over the last half of the integration, and
reproducible to within $\pm0.001$ under changes of the time step and of the
initial condition. This provides an independent and unambiguous verification of
deterministic chaos in the reduced collective-coordinate system.

\begin{figure}[t]
\centering
\begin{minipage}[t]{0.48\linewidth}
\centering
\includegraphics[width=\linewidth]{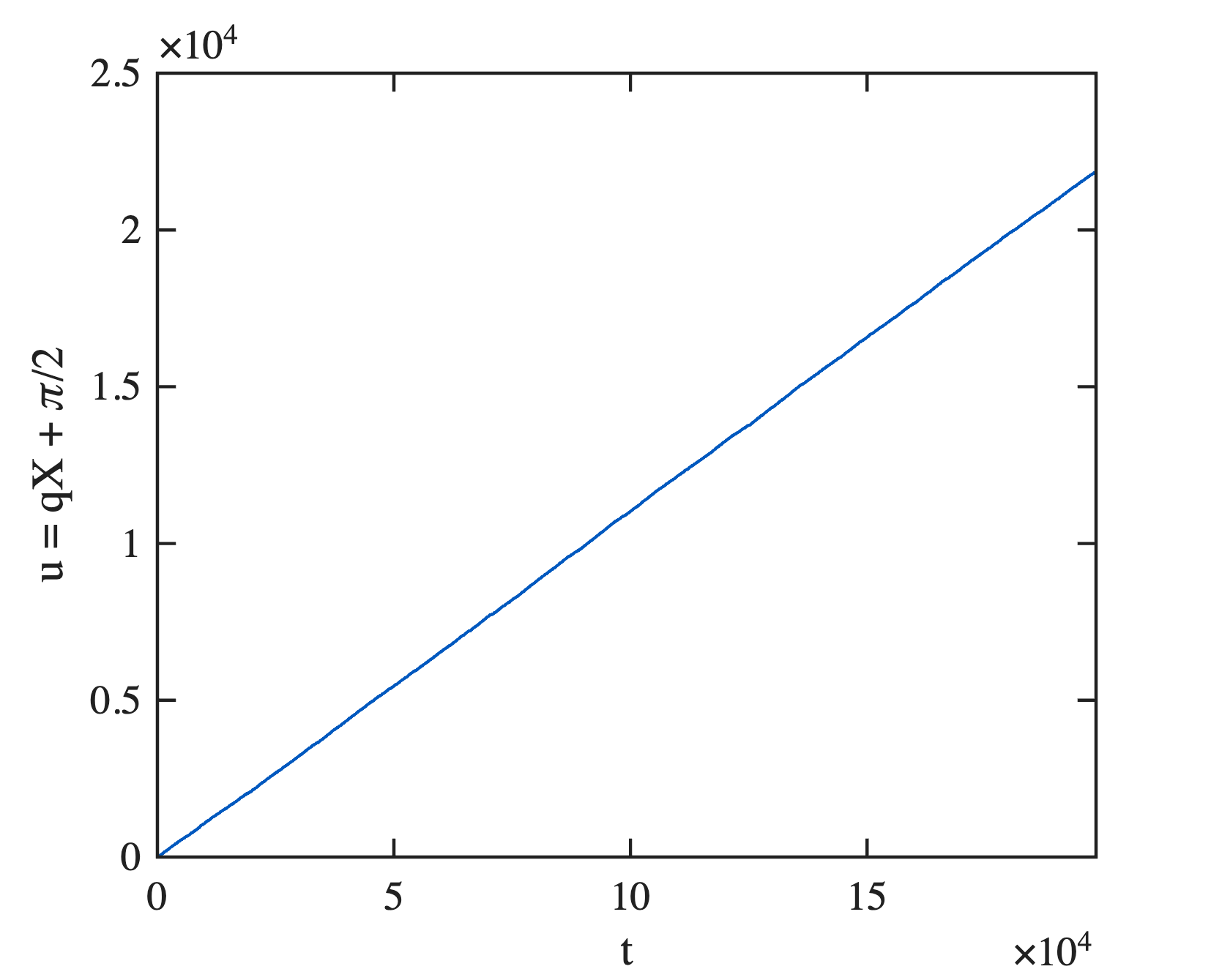}

\vspace{0.5em}
\small (a) Time evolution of the reduced collective-coordinate variable $u(t)$
for the same parameter values as in the Poincar\'e sections. The dynamics
exhibits a net drift corresponding to a depinned running state, modulated by
irregular fluctuations associated with chaotic motion.
\end{minipage}\hfill
\begin{minipage}[t]{0.48\linewidth}
\centering
\includegraphics[width=\linewidth]{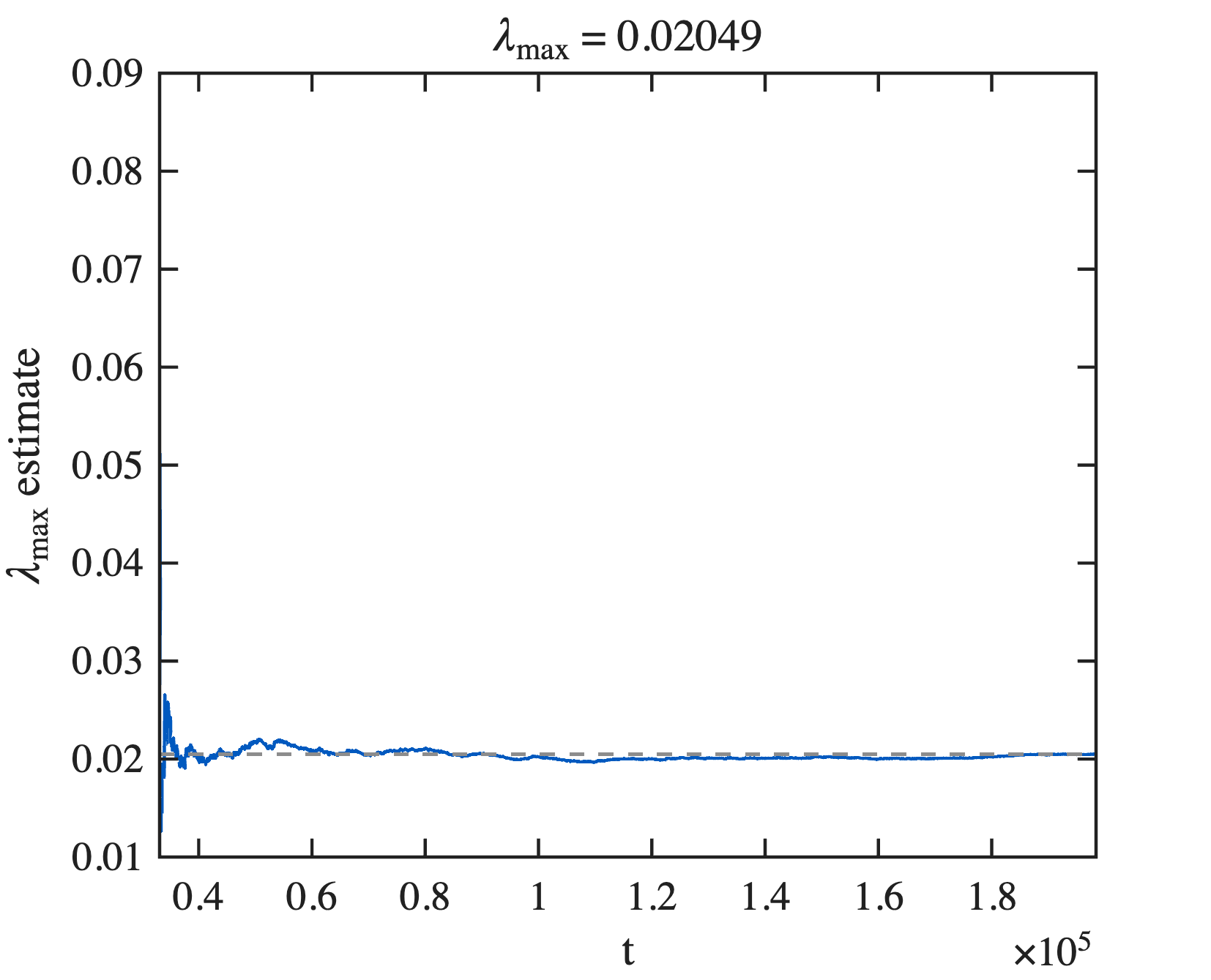}

\vspace{0.5em}
\small (b) Convergence of the largest Lyapunov exponent $\lambda_{\max}$ for the
reduced collective-coordinate system. After an initial transient, the exponent
converges to a strictly positive value, providing quantitative confirmation of
deterministic chaos.
\end{minipage}

\caption{Additional diagnostics of the reduced chaotic dynamics. Panel (a)
shows the time evolution of the collective-coordinate variable, while panel (b)
provides quantitative confirmation of chaos through the convergence of the
largest Lyapunov exponent.}
\label{fig:diagnostics_pair}
\end{figure}

\subsection*{Numerical methods}
\label{subsec:methods}

All results in this section concern the reduced collective-coordinate systems;
no simulation of the full partial differential equation is performed, and no
claim about the PDE is made on the basis of these computations. The quantities
reported are therefore properties of the finite-dimensional reductions
\eqref{eq:cc} and \eqref{eq:X_2cc}--\eqref{eq:a_2cc}, which is precisely the
setting in which the Melnikov analysis of Sec.~\ref{sec:melnikov} is valid.

The reduced systems are integrated with a fixed-step classical fourth-order
Runge--Kutta scheme. The step is fixed, and chosen as $\Delta t=T/n_{\rm sub}$
with $T=2\pi/\omega$ the forcing period and $n_{\rm sub}=400$, so that
stroboscopic samples fall exactly on multiples of $T$; an adaptive solver would
require interpolation and would degrade the phase accuracy of the Poincar\'e
map. The overlap coefficients $A_1$, $A_2$, $A_{2,a}$, $\alpha_1$, $\alpha_2$
are evaluated by composite Simpson quadrature on $\xi\in[-40,40]$ with $4\times
10^{4}$ intervals; since every integrand decays at least as $e^{-|\xi|}$, the
truncation error is below machine precision and the quadrature error is
$\sim10^{-11}$. All closed forms quoted in Sec.~\ref{sec_cc} were verified
against this quadrature.

The largest Lyapunov exponent is computed by Benettin's method: the variational
equation obtained by linearizing the reduced vector field about the trajectory
is integrated with the same Runge--Kutta stages, and the tangent vector is
renormalized once per forcing period, $\lambda_{\max}$ being the average of the
accumulated logarithmic stretching factors after the transient. For the 2-CC
system the tangent space is four-dimensional and the full Jacobian, including
the derivative of the back-reaction coefficient $\alpha(X,t)$ with respect to
$X$, is used. The value of $\lambda_{\max}$ reported below is stable to three
significant figures over the last half of the integration and reproducible to
within $\pm0.001$ under changes of $n_{\rm sub}$ between $300$ and $500$ and of
the initial condition.

\begin{remark}[Status of the Melnikov inequalities]
\label{rem:sufficient_only}
Throughout we write $\varepsilon_2^{\mathrm M}$, not
$\varepsilon_2^{\mathrm{crit}}$, for the amplitude at which the Melnikov function
\eqref{eq:Melnikov_two_contrib} first acquires simple zeros. The inequalities
\eqref{eq:chaos_threshold_k_eq_q} and \eqref{eq:chaos_criterion_2cc} are
sufficient conditions for transverse intersection of the perturbed stable and
unstable manifolds, and hence for the existence of a chaotic invariant set;
they are not exact critical values for the onset of observable chaos. Two
distinct gaps separate the two notions. First, a transverse homoclinic tangle
produces a Smale horseshoe, which in a dissipative system may be a chaotic
saddle rather than an attractor, so that typical trajectories need not exhibit a
positive Lyapunov exponent above $\varepsilon_2^{\mathrm M}$. Second, the
Melnikov function is only the first-order term in $\varepsilon$, so
$\varepsilon_2^{\mathrm M}$ is itself accurate only to leading order, with the
additional caveat of Remark~\ref{rem:beyond_all_orders} at large $\omega$. The
numerical results below are consistent with this reading and should be
interpreted accordingly.
\end{remark}

\subsection*{Testing the Melnikov threshold}

The results above establish that the reduced system is chaotic at one point in
parameter space lying above the analytical threshold. They do not yet test the
threshold itself. To do so we fix the driving frequency at $\omega=0.19$ and
sweep the forcing amplitude across the predicted value
$\varepsilon_2^{\mathrm{M}}=0.0541$, computing $\lambda_{\max}$ for each
$\varepsilon_2$ by the same procedure as above. The result is shown in
Fig.~\ref{fig:eps2_sweep}.

Below the threshold the outcome is unambiguous: at every sampled value of
$\varepsilon_2<\varepsilon_2^{\mathrm{M}}$ the exponent is strictly
negative, the trajectory settling onto a periodic attractor. The first positive
exponent appears at $\varepsilon_2\simeq0.067$, just above the analytical
prediction. This is a direct numerical confirmation of the criterion
\eqref{eq:chaos_threshold_k_eq_q}.

Above the threshold the picture is necessarily less clean, and it is worth
being precise about why. The Melnikov condition is sufficient for a transverse
homoclinic tangle, that is, for the existence of a chaotic \emph{invariant set}
in the form of a Smale horseshoe. In a dissipative system such a set need not
be an \emph{attractor}: it may be a chaotic saddle coexisting with, and
shadowed by, stable periodic orbits. Accordingly, Fig.~\ref{fig:eps2_sweep}
shows chaotic bands interleaved with periodic windows, and $\lambda_{\max}$ is
not positive at every $\varepsilon_2$ above threshold. The content of the test
lies in the other direction --- the complete absence of positive exponents
below $\varepsilon_2^{\mathrm{M}}$ --- which is exactly what a sufficient
criterion for separatrix splitting predicts, and which would fail if the
threshold were placed incorrectly.

\begin{figure}[t]
\centering
\includegraphics[width=0.62\linewidth]{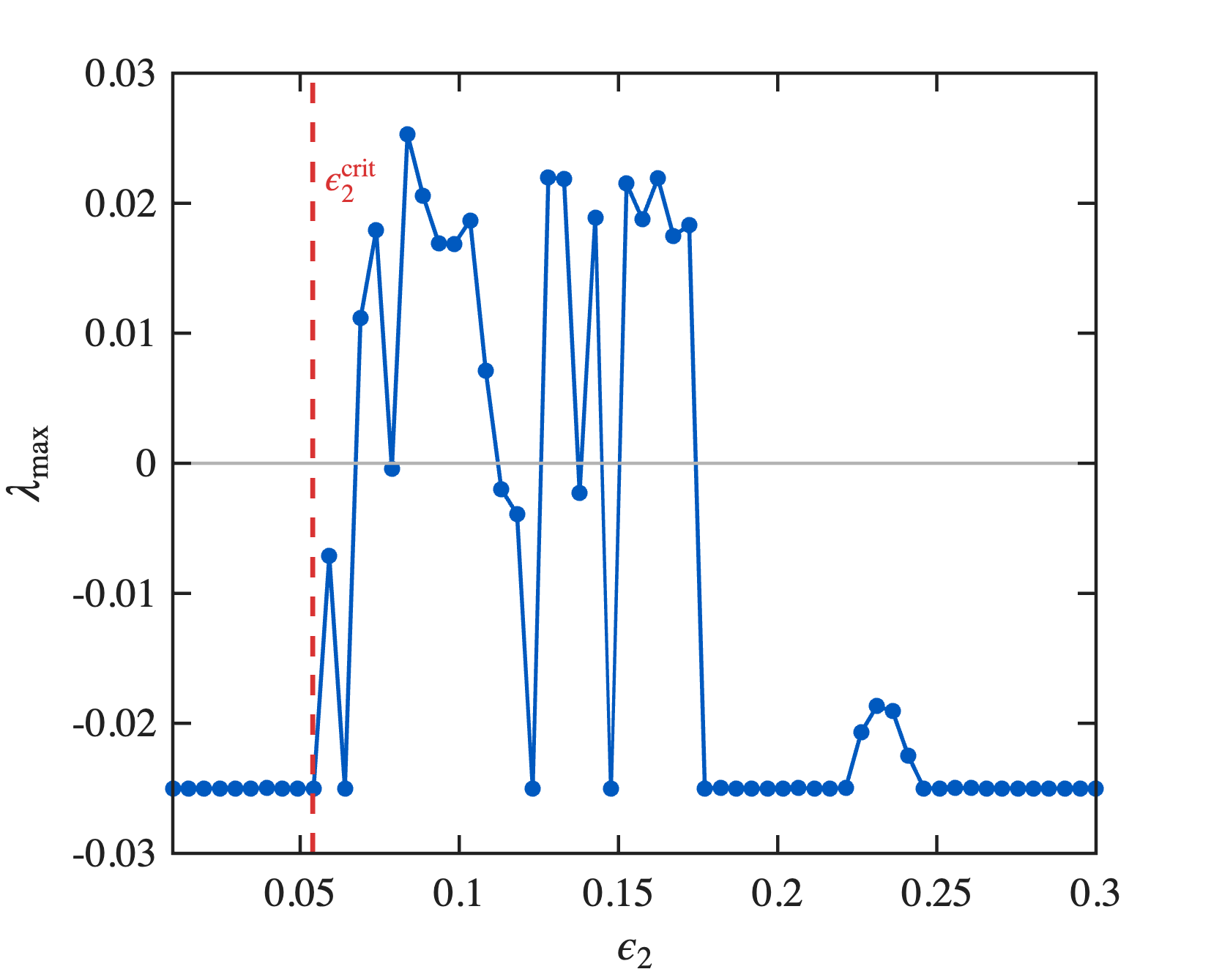}
\caption{Largest Lyapunov exponent of the reduced system as a function of the
forcing amplitude $\varepsilon_2$, at fixed $\omega=0.19$ with all remaining
parameters as in \eqref{eq:num_params}. The dashed line marks the analytical
threshold $\varepsilon_2^{\mathrm{M}}=0.0541$ predicted by
\eqref{eq:chaos_threshold_k_eq_q}. No positive exponent occurs anywhere below
the threshold; the first appears at $\varepsilon_2\simeq0.067$. Above the
threshold, chaotic bands alternate with periodic windows, as expected for a
dissipative system in which the Melnikov criterion guarantees a chaotic
invariant set but not a chaotic attractor.}
\label{fig:eps2_sweep}
\end{figure}

Outside the strict perturbative regime the analytical criteria should be
interpreted as sufficient rather than sharp thresholds. Nevertheless, the
numerical results show that the chaotic structures predicted by separatrix
splitting persist in a deformed form as parameters are increased, indicating
the robustness of the underlying geometric mechanism.

The Melnikov analysis developed in this work applies rigorously to the
finite-dimensional CC dynamics, where separatrix splitting and transverse
intersections of invariant manifolds can be characterized explicitly. It is
natural to ask how these predictions relate to the full partial differential
equation. One may in principle project PDE trajectories onto the kink manifold
by extracting the instantaneous kink center $X_{\mathrm{PDE}}(t)$ and its
velocity using the same diagnostic employed in the CC reduction, and sample
stroboscopically at the forcing period to obtain a Poincar\'e map in the reduced
variables $(u,v)=(qX,q\dot X)$. Such a comparison lies outside the scope of the
present work, whose object of study is the reduced phase space itself; we
comment here only on why the analytical framework does not transfer directly.

It is important to emphasize that the analytical notion of Melnikov chaos
employed in this work cannot be transferred in a straightforward manner to the
full $\phi^4$ field equation. In contrast to finite-dimensional Hamiltonian
systems, the driven--damped PDE considered here does not possess a well-defined
homoclinic orbit in phase space associated with a hyperbolic equilibrium, nor
does it admit an invariant separatrix structure whose persistence under
perturbations can be analyzed rigorously. As is well known from the theory of
homoclinic chaos in infinite-dimensional systems, a direct Melnikov formulation
at the PDE level typically requires additional structural assumptions, such as
periodic boundary conditions, the existence of an explicit homoclinic solution,
and a suitable spectral gap, as in the classical Melnikov analyses for nonlinear
Schr\"odinger-type equations \cite{LiMcLaughlin1994}.

The present $\phi^4$ model on the real line does not satisfy these
prerequisites. Accordingly, the absence of a direct Melnikov-type chaos
demonstration at the PDE level should not be interpreted as a limitation of the
analysis, but rather as a reflection of the fundamentally different dynamical
structure of the infinite-dimensional phase space. The collective-coordinate
reduction provides a natural and mathematically well-defined framework in which a
genuine homoclinic separatrix exists and can be perturbed systematically.
Within this reduced setting, Melnikov theory yields rigorous conditions for
separatrix splitting and chaotic dynamics, while the full PDE may display
additional effects such as radiation loss, basin structure, and strong pinning
that mask or suppress the reduced chaotic mechanism.   Together, the Poincar\'e sections and Lyapunov exponent
computations provide independent numerical confirmation
of the chaotic dynamics predicted by the Melnikov analysis.
\subsection*{Does the internal mode matter? A direct 1-CC versus 2-CC test}
\label{subsec:cc_compare}

The analysis of Sec.~\ref{subsec:extended_melnikov} predicts that the shape mode
is resonantly excited near $\omega\simeq\Omega_{\rm int}$ but that this
excitation does not feed back appreciably onto the translational splitting. Both
halves of this prediction are directly testable by integrating the 1-CC system
\eqref{eq:cc} and the 2-CC system \eqref{eq:X_2cc}--\eqref{eq:a_2cc} side by
side, from identical initial conditions, over a scan in $\omega$ that passes
through the internal-mode frequency. We do so at fixed $\varepsilon_2=0.15$
with the remaining parameters as in \eqref{eq:num_params}, integrating $1800$
forcing periods and discarding the first $600$.

Figure~\ref{fig:cc_compare}(b) confirms the first half. The r.m.s.\ internal
amplitude rises by a factor of about $19$ relative to its low-frequency value
and peaks at $\omega=1.215$, within $1\%$ of
$\Omega_{\rm int}=\sqrt{3/2}=1.2247$. The shape mode is therefore genuinely
resonantly driven, exactly as the linear response
\eqref{eq:a_forced_on_sep} predicts.

Figure~\ref{fig:cc_compare}(a) confirms the second half. The largest Lyapunov
exponents of the two models are indistinguishable across the entire scan,
including at resonance: the mean absolute difference is
$2\times10^{-4}$, two orders of magnitude below $\lambda_{\max}$ itself, and
comparable to the reproducibility of the exponent under changes of time step.
The two curves separate visibly at one frequency only, $\omega\simeq1.04$, where
the difference reaches $1.2\times10^{-2}$; there the two models settle on
different members of a set of coexisting attractors, a multistability effect
rather than a systematic shift of the threshold. Away from that isolated point
the internal mode leaves the translational chaos untouched.

Taken together, the two panels rule out the interpretation that the internal
mode is simply not excited. It is excited, strongly and resonantly; it merely
does not matter for the splitting, for the two structural reasons identified in
Sec.~\ref{subsec:extended_melnikov}. We note also that at
$\omega\simeq\Omega_{\rm int}$ both models are non-chaotic at this forcing
amplitude, which is consistent with Fig.~\ref{fig:melnikov_pair}(a): the
translational splitting there is exponentially small and
$\varepsilon_2^{\mathrm M}(\Omega_{\rm int})\simeq6$, far outside the
perturbative regime. Probing internal-mode-assisted chaos at resonance would
require $\Omega_0$ comparable to $\Omega_{\rm int}$, i.e.\
$\varepsilon_1\gtrsim3$ by \eqref{eq:Omega0_def}, at which point the
collective-coordinate reduction itself ceases to be justified. This is a genuine
structural limitation of the model, not of the analysis.

\begin{figure}[t]
\centering
\begin{minipage}[t]{0.48\linewidth}
\centering
\includegraphics[width=\linewidth]{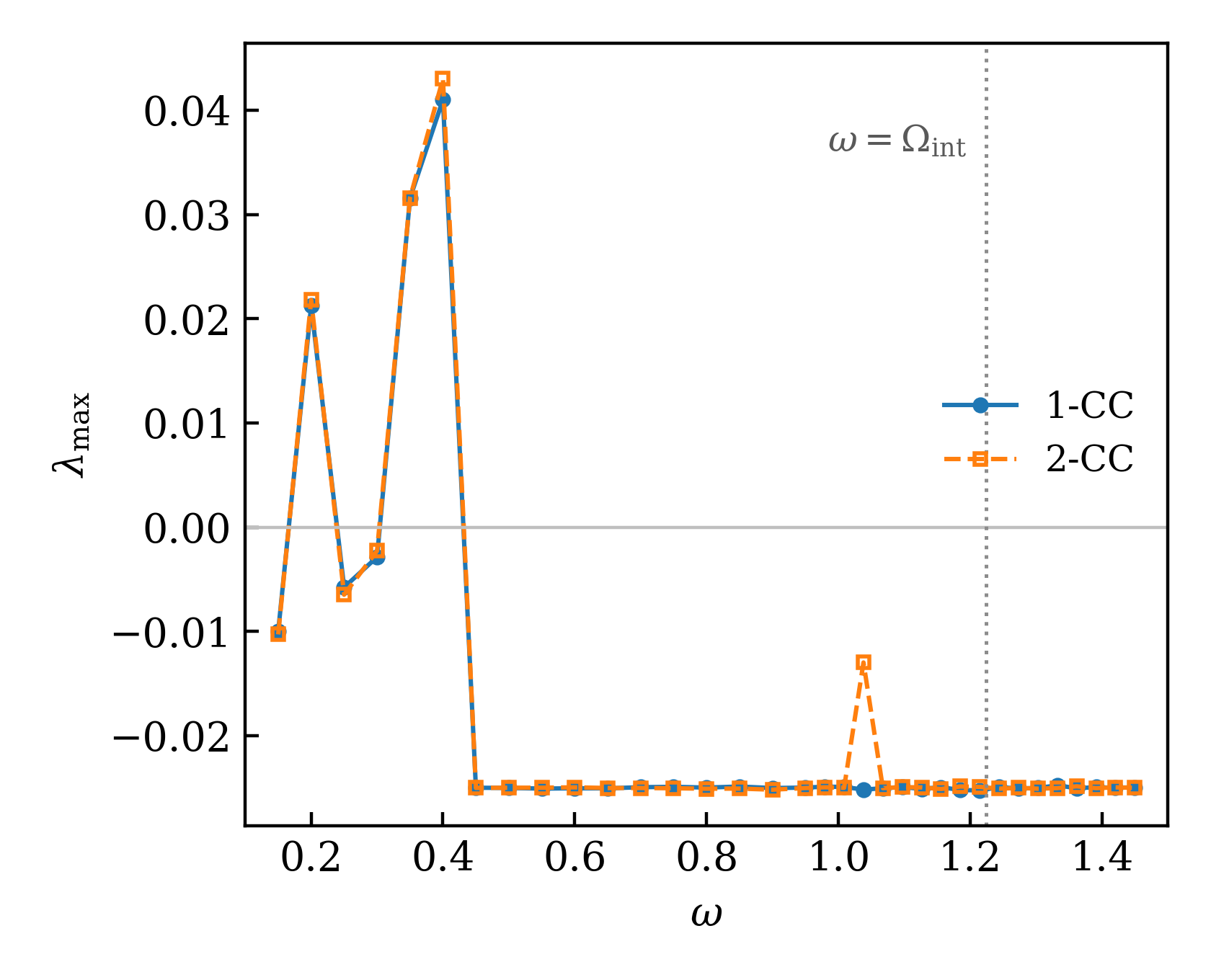}

\vspace{0.5em}
\small (a) Largest Lyapunov exponent of the 1-CC and 2-CC reduced systems versus
$\omega$ at $\varepsilon_2=0.15$. The curves are indistinguishable to within the
reproducibility of the exponent, except at the isolated point
$\omega\simeq1.04$ where the two models select different coexisting attractors.
\end{minipage}\hfill
\begin{minipage}[t]{0.48\linewidth}
\centering
\includegraphics[width=\linewidth]{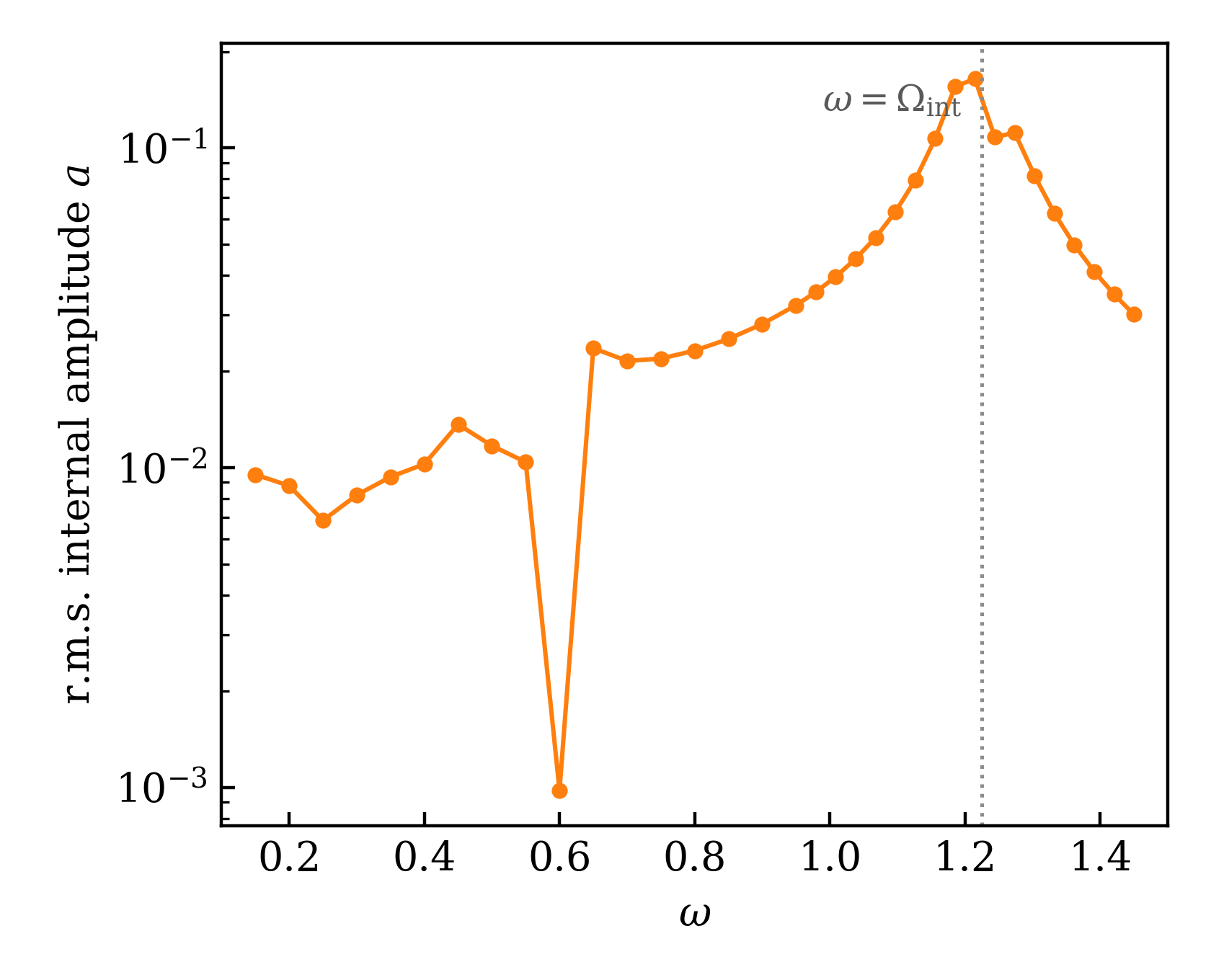}

\vspace{0.5em}
\small (b) Root-mean-square internal-mode amplitude in the 2-CC system over the
same scan. The response peaks at $\omega=1.215$, within $1\%$ of
$\Omega_{\rm int}=\sqrt{3/2}$, with an enhancement of about $19$ over its
low-frequency value.
\end{minipage}
\caption{Direct comparison of the translation-only and two--collective--coordinate
reductions across the internal-mode resonance. The shape mode is strongly and
resonantly excited (panel b) yet leaves the chaotic dynamics of the translational
coordinate unchanged (panel a), as predicted by the $\varepsilon_1$ ordering of
Lemma~\ref{lem:alpha_vanishes} together with the separatrix filtering of the
resonant response discussed after \eqref{eq:Cint}.}
\label{fig:cc_compare}
\end{figure}

\section{Conclusions}
\label{sec:conclusions}

We have analysed the mechanism by which chaotic dynamics arises in
collective-coordinate reductions of a driven and spatially modulated $\phi^4$
field, and we have delimited the range over which that mechanism operates. It is
worth summarising what the analysis establishes about the mechanism itself,
rather than restating the individual results.

Chaos in these reductions is produced by a single channel, which the analysis
identifies precisely. The spatial modulation $\mathcal F(x)$ does not merely
perturb the kink: by projecting onto the translational mode it manufactures an
effective pendulum whose pinning frequency
$\Omega_0=(\varepsilon_1A_1q/M)^{1/2}$ is fixed entirely by the overlap of the
modulation with the kink profile. That pendulum supplies a homoclinic orbit, and
it is the only homoclinic structure present in the reduced phase space. The
traveling wave $\lambda(t,x)$ then supplies a time-periodic drive whose overlap
with the separatrix velocity splits that orbit, while dissipation and bias
oppose the splitting and set the threshold. The chaos is therefore a property
neither of the kink nor of the drive separately, but of the interference between
two independent modulations, one spatial and one spatiotemporal, mediated by the
same kink profile: without $\varepsilon_1$ there is no separatrix to split, and
without $\varepsilon_2$ there is nothing to split it.

The natural expectation is that incorporating the shape mode opens a second
route to chaos, particularly near the internal-mode resonance. It does not, and
the reason is structural rather than numerical. We have proved
(Lemma~\ref{lem:alpha_vanishes}) that the leading back-reaction of the shape
mode on the kink centre vanishes identically, as a consequence of the
self-adjointness of the linearized operator together with the parity of its
translational zero mode. The shape mode can therefore influence the splitting
only at order $\varepsilon_1\varepsilon_2$, through the modulation-induced
coupling. A second and independent suppression then acts on what remains: the
resonant response of the shape mode builds up only where the forcing along the
separatrix has settled onto a single frequency, whereas the splitting integral
is weighted by $\dot X_0(t)$ and concentrated where it has not, so that the
resonant denominator is filtered out. The direct comparison of the two
reductions in Sec.~\ref{sec:numerics} confirms both halves of this picture: the
shape mode is resonantly excited, by a factor of about $19$ at
$\omega\simeq\Omega_{\rm int}$, and yet the Lyapunov exponents of the
translation-only and two--collective--coordinate models are indistinguishable.
Two further consequences of the same structure deserve record. In a homogeneous
medium ($\varepsilon_1=0$) the internal channel is absent altogether; and it
closes identically at the decoupling wavenumber $k^\ast=\sqrt{3/2}$, where the
traveling wave does not couple to the shape mode at all, with a near-closure
occurring in frequency at $\omega\simeq0.311$ through a cancellation in the
forced response across the separatrix passage.

The numerical work does more than illustrate these predictions. Sweeping the
forcing amplitude across the Melnikov value $\varepsilon_2^{\mathrm M}$ yields
no positive Lyapunov exponent anywhere below it, the first appearing immediately
above, which is the sharpest available confirmation that the analytical
criterion is correctly located. Above the threshold, chaotic bands alternate
with periodic windows, as must be expected of a criterion that guarantees a
chaotic invariant set rather than a chaotic attractor. The Poincar\'e sections,
in turn, show the chaotic set straddling the unperturbed separatrix and crossing
it repeatedly; in the unperturbed system this is impossible, since the
separatrix is then an invariant barrier between libration and rotation, so the
crossing is itself direct phase-space evidence of transverse splitting.

Three limitations of the analysis are intrinsic and should be stated plainly.
First, all results concern the finite-dimensional reductions. Melnikov theory
applies rigorously there, and we make no claim about the full field equation,
for the structural reasons set out in Sec.~\ref{sec:numerics}. Second, since
$\Omega_0=\mathcal O(\sqrt{\varepsilon_1})$, the splitting carries the factor
$\exp(-c\,\omega/\sqrt{\varepsilon_1})$ and is beyond all orders in the
rapid-forcing regime, where the first-order estimate requires
exponential-asymptotics control. Third, the internal-mode resonance is not
accessible perturbatively in this model: reaching
$\Omega_0\sim\Omega_{\rm int}$ would require $\varepsilon_1\gtrsim3$, at which
point the collective-coordinate reduction is itself no longer justified. The
last of these is a statement about the $\phi^4$ kink, whose shape-mode frequency
$\sqrt{3/2}$ is fixed, rather than about the method.

The approach developed here extends naturally to other nonlinear field theories,
to alternative spatiotemporal modulations, and to higher-dimensional
generalisations. The most informative extension, in our view, would be to a
model in which the pinning frequency and the internal-mode frequency can be
brought into resonance while both remain within the perturbative regime. The
machinery assembled above --- the constraint-consistent reduction, the
vanishing-back-reaction lemma, and the exact treatment of the internal response
along the separatrix --- would then predict a genuine internal-mode-assisted
lowering of the chaos threshold, which the present model cannot exhibit.

\bibliographystyle{amsplain}
\bibliography{f4_references_IJBC}

\end{document}